\documentclass[preprint,12pt]{elsarticle}

\usepackage{amssymb}
\usepackage{amsmath}
\usepackage{amsthm}
\usepackage{mathtools}
\usepackage{booktabs}

\theoremstyle{plain}
\newtheorem{theorem}{Theorem}
\newtheorem{lemma}[theorem]{Lemma}
\newtheorem{proposition}[theorem]{Proposition}

\newtheorem{example}{Example}

\theoremstyle{plain}
\newtheorem{definition}{Definition}

\theoremstyle{remark}

\newcommand{\K}{\mathcal{K}}

\newcommand{\R}{\mathbb{R}}
\newcommand{\Rn}{\mathbb{R}^{n}}
\renewcommand{\P}{\ensuremath{\mathbb{P}^{n}}}
\journal{Mathematical Social Sciences}

\begin{document}

\begin{frontmatter}



\title{A New Value for Cooperative Games on Intersection-Closed Systems}


\author{Martin \v{C}ern\'{y}} %% Author name

\affiliation{organization={Department of Applied Mathematics, Charles University},%Department and Organization
            %addressline={}, 
            city={Prague},
            %postcode={}, 
            %state={},
            country={Czech Republic}}

%% Abstract
\begin{abstract}
We introduce a new allocation rule, the \emph{uniform-dividend value} (UD-value), for cooperative games whose characteristic function is incomplete. The UD-value assigns payoffs by distributing the total surplus of each family of \emph{indistinguishable} coalitions uniformly among them. Our primary focus is on set systems that are \emph{intersection-closed}, for which we show the UD-value is uniquely determined and can be interpreted as the expected Shapley value over all positive (i.e., nonnegative-surplus) extensions of the incomplete game.

We compare the UD-value to two existing allocation rules for intersection-closed games: the R-value, defined as the Shapley value of a game that sets surplus of absent coalition values to zero, and the IC-value, tailored specifically for intersection-closed systems. We provide axiomatic characterizations of the UD-value motivated by characterizations of the IC-value and discuss further properties such as fairness and balanced contributions.  Further, our experiments suggest that the UD-value and the R-value typically lie closer to each other than either does to the IC-value.

Beyond intersection-closed systems, we find that while the UD-value is not always unique, a surprisingly large fraction of non-intersection-closed set systems still yield a unique UD-value, making it a practical choice in broader scenarios of incomplete cooperative games.
\end{abstract}



%% Keywords

\end{frontmatter}



\section{Introduction}
In a cooperative game of $n$ players, worth of cooperation of any subset of agents, a coalition $S$, is represented by a real value $v(S)$. The \emph{Shapley value} $\phi$ is a rule, which allocates the value of coalition of all agents, $v(N)$, among the agents in such a way that it reflects how the agents contribute to different coalitions. For each coalition, its \emph{surplus} $d_v(S)$ is computed and redistributed equally among its participants, i.e. each agent $i$ receives
\[
\phi_i(v) = \sum_{S \subseteq N:\ i \in S}\frac{d_v(S)}{|S|}.
\]
When not all of the coalitions are possible to form, a cooperative game is enhanced with a list of feasible coalitions $\mathcal{F}$ and at the same time, values of coalitions outside $\mathcal{F}$ are not provided as they are, due to their nature, non-existent or meaningless. 
In literature, these games are referred to as \emph{games with restricted cooperation} or \emph{cooperative games on set systems} (\cite{Myerson1977, Brink2017}). Computing the Shapley value in this modified setting rises a problem; a surplus of a coalition depends on the values of all of its subcoalitons, which might not lie in $\mathcal{F}$. A way around this problem was introduced in~\cite{Calvo2015}, where the surplus of unfeasible coalitions was set to zero. From the relation between surpluses of different coalitions, the rest of the surpluses can be uniquely determined and the Shapley value computed.

Another model where some of the coalition values are not provided are \emph{incomplete cooperative games} or \emph{partially defined cooperative games}~\cite{Masuya2016, Cerny2024}. Formally, the model is equivalent to games on set systems; a list of \emph{known coalitions} $\mathcal{K}$ is provided and the values for coalitions outside $\mathcal{K}$ are unknown. Although both models are formally equivalent, application of the \emph{R-value}, introduced above, is not always suitable under incomplete information as can be seen on the following example. Let $N = \{1,2,3\}$ be the set of players and the values being known for all coalitions but $\{1,2\}$:
\begin{itemize}
    \item $v(\{1\}) = v(\{2\}) = v(\{3\}) = 1$,
    \item $v(\{1,3\}) = v(\{2,3\}) = 3$,
    \item $v(\{1,2,3\}) = 7$.
\end{itemize}
For all coalitions, except $\{1,2\}$ and $\{1,2,3\}$, the surplus in uniquely defined. Under the R-value the surplus of $\{1,2\}$ is assumed to be zero, which results in division $(2.5,2.5,3)$. If the game was modeling a situation where the collaboration of players $1$ and $2$ was infeasible and at the same time, player $3$ can collaborate with either of the players, it seems reasonable to reward player $3$ with a higher payoff as is done under the R-value.

This division is less justified in situations, where collaboration of $\{1,2\}$ is feasible but $v(\{1,2\})$ is unknown. Under the R-value, it is implicitly assumed that there is no surplus of coalition $\{1,2\}$. THis means that, under the assumption of non-negative surpluses, the R-value corresponds to the most pessimistic scenario. 

One might argue, following the pattern of the game, that $v(\{1,2\}) = 3$ as is the case for other coalitions of size 2. Without additional assumptions, all we know about the surplus of $\{1,2\}$ is that it contributes to $v(\{1,2,3\})$ just as much as the surplus of  $\{1,2,3\}$ does. It is thus more fair to assume that both surpluses are equal.

This idea leads to the definition of our \emph{uniform-dividend value (UD-value)}. In the following section, we formally introduce this value and analyze when it is uniquely defined. Interestingly, uniqueness occurs when the set of known coalitions $\mathcal{K}$ is intersection-closed. For such systems, Béal et al.~\cite{Beal2020} recently introduced a value, which we refer to as the \emph{IC-value}. We also provide empirical analysis of uniqueness of the UD-value for set systems, which are not intersection-closed, however, our analysis restricts to intersection-closed set systems for the rest of this work.

Section~\ref{sec:characterization} is dedicated to characterizing the UD-value in terms of incomplete cooperative games. To illustrate, recall the example above. All the known surpluses are non-negative, and if we assume that the surpluses for $\{1,2\}$ and $\{1,2,3\}$ are also non-negative, then any assignment of values to these surpluses leads to what is called a \emph{positive extension} of the incomplete game. The expected Shapley value over all these possible extensions is exactly the UD-value. In Section~\ref{sec:characterization}, we formalize this intuition and highlight the correspondence between the UD-value and the \emph{average value} introduced in~\cite{Bok2023}.

We compare the UD-value with the IC-value and the R-value in Section~\ref{sec:comparison}. We show a common framework for defining all these values and show how it can be used to analyse them in the scope of incomplete cooperative games. We further modify two known axiomatizations of the IC-value provided in~\cite{Beal2020} to characterize both the UD-value and the R-value and illustrate difference between the values by providing examples of incomplete games, which violate axioms of three known characterizations of the R-value. We conclude the section with experimental analysis of the difference of the three values and their distance from uniform allocation rule.

\section{UD-value}\label{sec:ud-value}

%\begin{itemize}
%    \item hra
%    \item shapley value
%    \item unanimity game
%    \item positive game
%\end{itemize}
%Let us start with some preliminary definitions. A \emph{(complete) cooperative game}~\cite{Peleg2007} $(N,v)$ is given by set of players $N = \{1,\dots,n\}$ and a characteristic function $v \colon 2^N \to \R$, where $v(\emptyset)=0$. A \emph{unanimity game} $(N,u_S)$ for $\emptyset \neq S \subseteq S$ is defined as $u_S(T) = 1$ if $S \subseteq T$ and $u_S(T)=0$, otherwise. 

Recall a \emph{(complete) cooperative game} $(N,v)$, where $N$ is the player set and $v \colon 2^N \to \R$ with $v(\emptyset) = 0$ is the characteristic function. 
We call $S \subseteq N$ a \emph{coalition} and $v(S) \in \R$ represents the \emph{value of coalition $S$}.
%We call $x \in \Rn$ the \emph{payoff vector}, and denote $x(S) = \sum_{i \in S}x_i$. 
The \emph{Shapley value} $\phi(v) \in \Rn$ of a cooperative game $(N,v)$ is defined as $\phi_i(v) = \sum_{S \subseteq N, i \in S} \frac{d_v(S)}{|S|}$, where $d_v(S)$ is the \emph{surplus of coalition $S$ in $(N,v)$} and is defined recursively as $d_v(S) = v(S) - \sum_{T \subsetneq S}d_v(T)$ with $d_v(\{i\}) = v(\{i\})$. 
When $d_v(S) \geq 0$ for every coalition $S$, we say $(N,v)$ is \emph{positive} and denote the set of all positive games on $n$ players by $\mathbb{P}^n$.

We write $v$ instead of $(N,v)$ when $N$ is known from the context. For the sake of brevity, we write $i$ instead of $\{i\}$, and %$i,j$ instead of $\{i,j\}$, etc.
by $n,s,t$, we denote the sizes of coalitions $N$, $S$, and $T$, respectively.

%Mozna budu potrebovat i superadditive a convex?

\begin{definition}\label{def:incomplete-game}
    An \emph{incomplete cooperative game} $(N,\K,v)$ is given by a set of agents $N = \{1,\dots,n\}$, a set of known coalitions $\mathcal{K} \subseteq 2^N$ and a characteristic function $v \colon 2^N \to\R$. Further, $\emptyset \in \K$ and $v(\emptyset)=0$.
\end{definition}
When $\K = 2^N$, an incomplete cooperative game $(N,\K,v)$ coincides with a complete cooperative game $(N,v)$. In Definition~\ref{def:incomplete-game}, $\K$ serves as a “mask” over $v$, revealing only those coalition values that are known. Alternatively, one can define incomplete game as $v \colon \K \to \mathbb{R}$ (see e.g.~\cite{Masuya2016, Bok2023}). Although in many situations both viewpoints are equivalent, it can sometimes be advantageous to switch between them.

In example provided in the introduction, coalitions $\{1,2\}$ and $\{1,2,3\}$ were \emph{indistinguishable}, because the sum of their surpluses was affecting the same known value, $v(N)$. For a more general structure of known coalitions $\K$, a coalition $T$ might affect more than one value. From the definition of surplus, $T$ affects the value of $S$, if $T$ is a subset of $S$. Therefore, two coalitions are in general \emph{indistinguishable}, if they are subsets of the same system of known coalitions. Formally, we employ the \emph{closure of coalition $T$ in a set system $\K$} defined as
$$c_\K(T) = \bigcap_{S \in \K, T \subseteq S}S.$$
Every $X,Y \subseteq N$ such that $c_\K(X) = c_\K(Y)$ correspond to indistinguishable coalitions given by our motivation. We denote $\mathcal{C}(\K) = \{c_\K(S) \mid S \subseteq N\}$ and $\mathcal{C}(S) = \{ T \subseteq N \mid c_\K(T) = c_\K(S)\}$. Function
$c_\K \colon 2^N \to 2^N$ is a \emph{closure operator on $N$}, for which there are classical known results.
\begin{lemma}\cite{Davey2002}\label{lem:operator}
    Let $\K \subseteq 2^N$ be a set of known coalitions. Then it holds
    \begin{enumerate}
        \item $T \subseteq S \implies c_\K(T) \subseteq c_\K(S)$,
        \item $S \in \K \implies c_\K(S) = S$.
    \end{enumerate}
\end{lemma}
%We write $c(S)$ instead of $c_\K(S)$ whenever $\K$ is clear from the context.

The key idea behind defining our value is to set the surpluses of indistinguishable coalitions to be equal. Since the surpluses are in some sources referred to as dividends, we call it the \emph{uniform-dividend value}.

\begin{definition}\label{def:UD-value}
    The \emph{uniform-dividend value (UD-value)} $\Phi^\K(v)$ of an incomplete game $(N,\K,v)$ is defined for $i \in N$ as
    $$\Phi_i^\K(v) = \sum_{S \subseteq N, i \in S}\frac{\delta_v^\K(S)}{|S|},$$
    where $\delta_v^\K(S)$ for $S \subseteq N$ are given by the following conditions:
    \begin{enumerate}
        \item $\sum_{T \subseteq S}\delta_v^\K(T) = v(S)$ for $S \in \K$,
        \item $\delta_v^\K(S) = \delta_v^\K(T)$ if $c_\K(S) = c_\K(T)$ for every $S,T \subseteq N$.
    \end{enumerate}
\end{definition}

To compute values $\delta_v^\mathcal{K}$, one might solve the given system of linear equations. It is not immediately clear, though, whether these systems define a unique solution. When $\K = 2^N$, there are no type 2.\ equations, and type 1.\ equations now correspond to conditions on coalition surpluses, thus $\delta_v^\mathcal{K}(S) = d_v(S)$ for every coalition $S$. Conversely, when $\mathcal{K} = \{\emptyset, N\}$, there is only one type 1.\ equation, while all nonempty subsets $S$ and $T$ satisfy $c_\K(S) = c_\K(T)$. As the size of the set $\mathcal{K}$ increases, the number of type 1.\ equations grows, while the number of type 2.\ equations diminishes. An important case, when the system of linear equations comes into balance is when $\K$ is \emph{intersection-closed}.

\begin{definition}
    An incomplete cooperative game $(N,\K,v)$ is \emph{intersection-closed} if
    $\forall S,T \in \K \text{ it holds } S \cap T \in \K$.
\end{definition}

The reason behind why the UD-value is uniquely defined for intersection-closed games stems from the fact that $\mathcal{C}(\K) = \K$, which means for every $T \subseteq N$, $c_\K(T) \in \K$.

\begin{proposition}\label{prop:ic_implies_unique}
    The UD-value is unique if $(N,\K,v)$ is intersection-closed game with $N \in \K$.
\end{proposition}
\begin{proof}
    For $S \in \K$, consider the condition
    \begin{equation}\label{eq:1}
        \sum_{T \subseteq S}\delta_v^\K(T) = v(S).
    \end{equation}
    In~\eqref{eq:1}, one can substitute each $\delta_v^\K(T)$ with $\delta_v^\K(c_\K(T))$, since these are, according to Definition~\ref{def:UD-value}, equal. Since $T \subseteq S$, then $c_\K(T) \subseteq c_\K(S) = S$, where the equality follows from the fact that $c_\K(S) \in \K$. One can thus rewrite Eq.~\eqref{eq:1} as
    \begin{equation}
        \sum_{T \in \K, T \subseteq S}\alpha_T \delta_v^\K(T) = v(S)
    \end{equation}
    where $\alpha_T = |\mathcal{C}(T)|$, because $X \in \mathcal{C}(T) \implies X \subseteq T$. 

    Now assume $S \in \K$ is inclusion-minimal set in $\K$. From our previous analysis, it follows
    \begin{equation}
        \sum_{T \subseteq S}\delta_v^\K(T) = |\mathcal{C}(S)|\delta_v^\K(S) = v(S),
    \end{equation}
    thus $\delta_v^\K(S)$ is uniquely defined. 
    
    The rest of the values can be iteratively determined by considering the condition for the inclusion-minimal coalition $S \in \K$, for which $\delta_v^\K(S)$ is not yet determined. This is because it holds
    \begin{equation}\label{eq:3}
        \sum_{T \subseteq S}\delta_v^\K(T) = \sum_{T \in \K, T \subsetneq S}|\mathcal{C}(T)|\delta_v^\K(T) + |\mathcal{C}(S)|\delta_v^\K(S) = v(S).
    \end{equation}
    Since $S$ is inclusion-minimal coalition for which $\delta_v^\K(S)$ is not yet determined, it can be derived from Eq.~\eqref{eq:3}.

\end{proof}
When the game is not intersection-closed, type 1.\ conditions from Definition~\ref{def:UD-value} can be rewritten, similarly to the proof of Proposition~\ref{prop:ic_implies_unique}, as
\begin{equation}\label{eq:2}
    \sum_{T \in \mathcal{C}(\K), T \subseteq S}\delta_v^\K(T) = v(S),\hspace{2ex} S \in \K.
\end{equation}
These constitute a system of $|\K|$ equations and $|\mathcal{C}(\K)|$ variables. For games, which are not intersection-closed, $|\mathcal{C}(\K)| > |\K|$, thus the values $\delta_v^\K$ form a non-trivial affine space. This does not immediately imply that the UD-value is not unique, as one could potentially get the same Shapley value for different values $\delta_v^\K$. We performed a numerical analysis, which not only shows that this equality of Shapley values for different dividends happen, but as $n$ grows larger, the proportion of set systems, which are not intersection-closed but yield a unique UD-value grows. In our analysis, we assume $N \in \K$, as when this is not the case, the UD-value is never unique and can be arbitrarily large for every player. Our analysis revealed that for $n=2$ and $n=3$, the UD-value is unique only when the game is intersection-closed. For $n=4$, we found 15 set systems, which are not intersection-closed, and still, for every feasible system of $\delta_v^\K$, the UD-value is the same for every incomplete game. These set systems follow a very strict and very symmetrical structure $\K = \{S \subseteq N \mid |S| = 2\} \cup \{\emptyset, N\} \cup \mathcal{S}$ where $\mathcal{S}$ is a nonempty subset of singleton coalition.
For $n=5$ and $n=6$, even less symmetric set systems can be found. Because the number of different set systems is double exponential in $n$, we approximate the proportion of systems with unique UD-value by sampling random set systems and checking the uniqueness on this sample. To determine the sample size, we use the sample size formula~\cite{Yamane1967},
\begin{equation}
    n = \frac{Z^2\cdot p \cdot (1-p)}{E^2}
\end{equation}
where $n$ is the sample size, $Z$ is the Z-score corresponding to the desired confidence level, $p$ corresponding to the estimated proportion of the population and $E$ corresponding to the margin of error. We set $Z = 2.576$ (which corresponds to 99\% confidence level), $p = 0.5$ (since this is the most reserved estimate and yields the highest $n$) and $E = 0.001$, which results in $n=1,658,944$. At the end, we selected a larger sample size of $n= 2,000,000$. The resulting estimated proportions can be found in Table~\ref{tab:set_systems}.
The experiments shows that the number of intersection-closed set systems decrease rapidly, as for 4 players, the number is $2,271$, which corresponds roughly to $14\%$, while for 5 players the proportion is $0.1\%$ and for 6 players, among $2,000,000$ randomly sampled set systems, there was no intersection-closed one. The proportion of set systems, which are not intersection-closed, but yield unique UD-value grows with $n$; for 5 players, it is roughly $23\%$ and for 6 players, it is even around $72\%$. This shows that for $n$ large enough, the UD-value is suitable for increasingly more and more set systems.
\begin{table}[ht!]
\centering
\begin{tabular}{@{}lcccc@{}}
\toprule
\textbf{Set Systems} & \textbf{\(|N| = 3\)} & \textbf{\(|N| = 4\)} & \textbf{\(|N| = 5\)} & \textbf{\(|N| = 6\)} \\ \midrule
Containing $\emptyset, N$ & 64 & 16,384 & 1,073,741,824 & $\sim 8.507 \cdot 10^{37}$ \\
I-C systems & 0.70313 & $0.13861$ & $\sim1.28168 \cdot 10^{-3}$ & $\sim 0.0$ \\
Unique UD & 0 & $9.15527 \cdot 10^{-4}$ & $\sim0.22874$ & $\sim0.72484$ \\ \bottomrule
\end{tabular}
\caption{Comparison of set systems for ground set sizes \( |N| = 3, 4, 5, 6 \).}
\label{tab:set_systems}
\end{table}

Despite these results concerning uniqueness of the UD-value, we restrict our analysis in the rest of this text only to intersection-closed set systems with $N \in \K$; strong theoretical properties allow for stronger results, we can compare the UD-value with the IC-value by Beál et al. and for $n=3$, $4$, intersection-closed set systems cover most scenarios where the UD-value is unique.

% \begin{example}
% Let $(N,\K,v)$ be given by $N = \{1,2,3\}$ and $\K = 2^N \setminus \{1,2\}$ and further $v(S) = 0$ for every $S \in \K$, $S \neq N$ and $v(N) = 1$.

% The R-value splits $1$ equally between the agents, however, the UD-value is assigned in the following way:
% \begin{itemize}
% \item $\Phi_1(v,\K) = \frac{5}{12}$,
% \item $\Phi_2(v,\K) = \frac{5}{12}$,
% \item $\Phi_3(v,\K) = \frac{2}{12} = \frac{1}{6}$.
% \end{itemize}
% The emphasis of 
% \end{example}
% \begin{proposition}
% If $(N,\K,v)$ is not intersection-closed, the R-value might not be among UD-values.
% \end{proposition}
% \begin{proof}
% I found an example for $N = \{1,2,3\}$ and $\K = \{\emptyset, \{1,2\}, \{2,3\}, \{1,2,3\}\}$. The R-value of this game should correspond to
% \begin{itemize}
% \item $R_1 = \frac{v(12)}{2} + \frac{v(123)}{3}$,
% \item $R_2 = \frac{v(12)+v(23)}{2} + \frac{v(123)}{3}$,
% \item $R_3 = \frac{v(23)}{2} + \frac{v(123)}{3}$.
% \end{itemize}
% The UD-value of this game should correspond to
% \begin{itemize}
% \item $\Phi_1 = \frac{3}{2} \delta_1 + \frac{5}{6} \delta_N$,
% \item $\Phi_2 = \delta_2 + \frac{\delta_1 + \delta_3}{2} + \frac{\delta_N}{3}$,
% \item $\Phi_3 = -\frac{3}{2} \delta_3 + \frac{5}{6} \delta_N$.
% \end{itemize}
% Now if I try to find $\delta$'s by solving a system of linear equations with $R_i = \Phi_i$, the system is unsolvable for some values of $v(12),v(23),v(123)$.
% \end{proof}

\section{Characterization of the UD-value}\label{sec:characterization}
When we assume that a complete game must satisfy positivity (i.e., non-negative surpluses), the known values of certain coalitions impose constraints on the unknown values. For example, if $S$ is a subset of $T$, then $v(S)$ forms a lower bound for $v(T)$ and vice versa. By combining the known values of multiple coalitions, one can derive even tighter bounds for other (unknown) coalitions. Moreover, certain selections of values for the unknown coalitions may conflict with the positivity requirement if they fail to respect these bounds.

All of this information—both about feasible ranges of values and about which combinations of unknown values are disallowed—can be compactly represented by the set of all extensions. An extension is simply a complete game that agrees with the known values and remains consistent with positivity. Thus, identifying the set of all extensions precisely captures everything we can infer about the unknown values under these conditions.
\begin{definition}
    A cooperative game $(N,w)$ is a \emph{$\P$-extension} of incomplete cooperative game $(N,\K,v)$ if $(N,w) \in \P$ and for every $S\in \K$,
    \begin{equation}
        w(S)=v(S).
    \end{equation}
    Further, $(N,\K,v)$ is \emph{\P-extendable}, if it has a \P-extension.
\end{definition}
Without additional assumptions on the incomplete game, every \(\P\)-extension is equally likely to represent the actual underlying game. Hence, if our goal is to predict the Shapley value of that underlying game but we only know partial information, a natural approach is to average the Shapley values across all possible \(\P\)-extensions. This expected Shapley value is then our best guess, given the incomplete data. Our main result shows that this expected payoff vector coincides with the UD-value. In other words, the UD-value for an incomplete game can be interpreted precisely as the average over all Shapley values of its \(\P\)-extensions.

\begin{theorem}\label{thm:char}
For \P-extendable intersection-closed game $(N,\K,v)$, it holds
\begin{equation}
    \Phi^\K(v) = \mathbb{E}_{w \sim \P(v)} \left[\phi(w)\right].
\end{equation}
\end{theorem}

Before proving this result, we need to describe the structure of the set of \(\P\)-extensions. In particular, we show that it can be decomposed into a family of simplices, each associated with a set of indistinguishable coalitions $\mathcal{C}(S)$. Concretely, the surpluses of the coalitions in $\mathcal{C}(S)$ can be viewed as points in a regular simplex, where the coordinates (i.e., the surplus values) sum to a special constant $\Delta_v(S)$. Each vertex of this simplex corresponds to allocating the entire amount $\Delta_v(S)$ to exactly one coalition $T \in \mathcal{C}(S)$, while assigning zero surplus to all other coalitions in $\mathcal{C}(S)$.

\begin{proposition}\label{prop:p-extensions}
    Let $(N,\K,v)$ be a \P-extendable intersection-closed game and let $\Delta_v(S) = v(S) - \sum_{T \in \K, T \subsetneq S}\Delta_v(T)$\footnote{Notice the resemblance of $\Delta_v(S)$ with the definition of the surplus $\delta_v(S)$.}. The set of \P-extensions $(N,w)$ can be expressed as
    \begin{equation}
    \left\{(N,w) \middle|\ \forall S \in \K: \sum_{T \in \mathcal{C}(S)}d_w(T) = \Delta_v(S) \text{ and } d_w(T) \geq 0, \forall T \subseteq N\right\}.
\end{equation}

\end{proposition}
\begin{proof}
    Once again, we use the fact from Lemma~\ref{lem:operator}, which states $X \subseteq T \implies c_\K(X) \subseteq c_\K(T)$ for every $X \subseteq T \subseteq N$. From this property, for any \P-extension $(N,w)$ and any $S \in \K$, we have
    \begin{equation}\label{eq:4}
        v(S) = w(S) = \sum_{T \subseteq S}d_w(S) = \sum_{ T \in \K, T \subseteq S}\sum_{X \in \mathcal{C}(T)}d_w(X).
    \end{equation}
    Now by fixing $\Delta_v(T) = \sum_{X \in \mathcal{C}(T)}d_w(X)$, we get from Eq.~\eqref{eq:4} that
    \begin{equation}
        \Delta_v(S) = v(S) - \sum_{T \subseteq S, T \in \K}\Delta_v(T).
    \end{equation}
    From the combination with non-negativity of surpluses, the result follows.
\end{proof}

Another way to interpret Proposition~\ref{prop:p-extensions} is to note that, we know that for each set of indistinguishable coalitions $\mathcal{C}(S)$, the positivity assumption fixes the total surplus $\Delta_v(S)$ but does not specify how it is distributed among the individual coalitions. Hence, any way of dividing $\Delta_v(S)$ among the coalitions in $\mathcal{C}(S)$ is equally plausible. In particular, if we take the average of all such distributions, each coalition in $\mathcal{C}(S)$ ends up with the same expected surplus. This uniform allocation across indistinguishable coalitions is exactly the principle behind the UD-value, illustrating its natural interpretation as an expected or average allocation when no further information is available.

% Alternative way to view the result from Proposition~\ref{prop:p-extensions} is that we know from the incomplete game and the positive assumption that the total surplus of all coalitions from one set of indistinguishible coalitions $\mathcal{C}(S)$ is given, and it is the value $\Delta_v(S)$. What remains unknown about the game, though, is the distribution of this value among surpluses of all coalitions $T \in \mathcal{C}(S)$. It is equally likely that all the surplus $\Delta_v(S)$ is due to one coalition in $\mathcal{C}(S)$ as well as any other distribution of this value among all the indistinguishable coalitions. The expected surplus corresponds to the centre of the simplex, which at the same time corresponds to equal surpluses of all of the coalitions. Equal surpluses of all of the indistinguishable coalitions is assumed for the UD-value.

\begin{proof}[Proof of Theorem 5]
    From linearity of the Shapley value and the expectation, we have
    \begin{equation*}
        \mathbb{E}_{w \sim \P(v)} \left[\phi(w)\right] = \phi(\mathbb{E}_{w \sim \P(v)} \left[w\right]).
    \end{equation*}
    Since $\P(v)$ is a combination of simplexes, it follows the average is given by an average of its vertices, i.e., $d_{\mathbb{E}(w)}(S) = \frac{\Delta_{c_\K(S)}}{|\mathcal{C}(c_\K(S))|}$. As it holds
    $d_{\mathbb{E}(w)}(S) = d_{\mathbb{E}(w)}(T)$ for every $S$, $T \subseteq N$ with $c_\K(S) = c_\K(T)$, we conclude
    \begin{equation*}
        \phi(\mathbb{E}_{w \sim \P(v)}(w)) =  \Phi^\K(v).
    \end{equation*}
\end{proof}

Theorem~\ref{thm:char} establishes a correspondence between the UD-value and the \emph{average value}~\cite{Bok2023}, originally studied in a different context. Despite differences in the underlying assumptions, both results share the same core principle: they derive a value by averaging across all admissible extensions of the game.

In Theorem~\ref{thm:char}, we implicitly assume that \((N, \K, v)\) is \(\P\)-extendable, without detailing the precise conditions under which this assumption holds. The following result clarifies the situation by drawing a neat link between \(\P\)-extendability for intersection-closed incomplete games and the positivity condition for complete games. Specifically, as the non-negativity of surpluses \(d_v(S)\)  imply positivity of complete games, the non-negativity of surpluses \(\Delta_v(S)\) imply \P-extendability of incomplete game.

\begin{proposition}
Let $(N,\K,v)$ be an intersection-closed incomplete game. It is $\P$-extendable if and only if $\Delta_v(S) \geq 0$ for every $S \in \K$.
\end{proposition}
\begin{proof}
    When there is a \P-extension $(N,w)$, it is immediate 
    \[\Delta_v(S) = \sum_{T \in \mathcal{C}(S)}d_w(T) \geq 0.\]
    Conversely, if $\Delta(S) \geq 0$ for every $S \in \K$, one can contruct $(N,v_\delta)$ through its surpluses for every $T \subseteq N$ as
    \[
    d_{v_\delta}(T) = \frac{\Delta(c_\K(T))}{|\mathcal{C}(T)|} \geq 0.
    \]
    From the construction, it follows $(N,v_\delta)$ is a \P-extension of $(N,\K,v)$.
\end{proof}

% Old proof for different proposition
% \begin{proof}
%     Since $(N,v_\delta)$ is an extension of $(N,\K,v)$, its positivity implies $\P$-extendability. If, on the other hand, there is a \P-extension $(N,w)$, it is not difficult to see that transfering $\varepsilon$ from the surplus of $S$ to the surplus of $T$ results in an extension of $(N,\K,v)$ if $c(S) = c(T)$. The game $(N,v_\delta)$ can be constructed from $(N,w)$ by taking the average of dividends over sets of indistinguishable coalitions. Non-negativity of surpluses of $(N,w)$ thus imply non-negativity of surpluses of $(N,v_\delta)$.
% \end{proof}
% There is alse a nice correspondence between positivity of a complete game $(N,w)$, which is given be $d_w(S)\geq 0$ for every $S \subseteq N$ and \P-extendability of $(N,\K,v)$ which is given by $\Delta_v(S) \geq 0$ for every $S \in \K$.

The complete extension $(N, v_\delta)$ introduced in the proof of the previous result can also serve as a building block for defining the UD-value. In the next section, we illustrate how this (and other) complete extensions can be used to formally define the three values discussed thus far.

\section{Comparison with Other Values}\label{sec:comparison}

In the introduction, we motivated the UD-value by illustrating a situation in which the R-value may not be the most suitable option. Subsequently, we showed that, unlike the R-value, the UD-value need not be unique. However, for intersection-closed systems, which we focus on, the UD-value is unique. In a recent work, Béal et~al.~\cite{Beal2020} introduced a value for intersection-closed systems—based on the Shapley value—which we refer to as the \emph{IC-value}.

In this section, we compare the UD-value with both the R-value and the IC-value. We begin by formally defining all three values (UD, R, and IC) using the same framework of special extension. We then demonstrate how this framework can be used to compare the three values from the perspective of incomplete games.

Béal et~al.~\cite{Beal2020} present two characterizations of the IC-value, and they note that both can be adapted to characterize other values. Here, we provide a detailed derivation of the two characterizations for the UD-value and sketch how a similar adaptation can be carried out for the R-value. Furthermore, we contrast these characterizations with the known axiomatizations of the R-value from~\cite{Calvo2015, Albizuri2022}, showing that each of the three known axiomatizations contains at least one axiom the UD-value does not satisfy.

Finally, using numerical experiments, we analyze additional differences in how these three values behave and assess which value distributes payoffs most uniformly among the agents.

Throughout this section, we denote the UD-value, the R-value, and the IC-value of an incomplete game \(\bigl(N,\K,v\bigr)\) by \(UD^\K(v)\), \(R^\K(v)\), and \(IC^\K(v)\), respectively. We often omit \(\K\) in the notation whenever this does not cause confusion.

\subsection{Unified framework}

The UD-value, the R-value, and the IC-value are all motivated by the Shapley value $\phi$. The approach all these three values take in dealing with the fact that the Shapley value is defined only for complete games is to implicitly, or explicitly assume a complete extension of the incomplete game.
\begin{definition}\label{def:special-games}
    Let $(N,\K,v)$ be an intersection-closed incomplete game. The R-game, the IC-game, and the UD-game $(N,v_{R})$, $(N,v_{IC})$, $(N,v_{UD})$ are defined as follows:
    \begin{enumerate}
        \item $d_{v_R}(S) = \begin{cases}
            v(S) - \sum_{T \subseteq S} d_{v_{R}}(T) & S \in \K,\\
            0 & S \notin \K.\\
        \end{cases}$
        \item $v_{IC}(S) = v(c_\K(S))$ for $S \subseteq N$.
        \item $d_{v_{UD}}(S) = \delta_v^\K(S)$ for $S \subseteq N$, where $\delta_v^\K(S)$ is from Definition~\ref{def:UD-value}.
    \end{enumerate}
\end{definition}
Shapley value of each of the special complete games now determines one of the values.
\begin{definition}\label{def:special-values}
    Let $(N,\K,v)$ be an intersection-closed incomplete game. The \emph{R-value} $R^\K(v)$, the \emph{IC-value} $IC^\K(v)$, and the \emph{UD-value} $UD^\K(v)$ can be defined as follows:
    \begin{enumerate}
        \item $R^\K(v) = \phi(v_R)$,
        \item $IC^\K(v) = \phi(v_{IC})$,
        \item $UD^\K(v) = \phi(v_{UD})$.
    \end{enumerate}
\end{definition}
Games from Definition~\ref{def:special-games} are extensions of $(N,\K,v)$. Following the scope of incomplete games, we can compare the values by studying properties of these extensions. We show that games $(N,v_R)$ and $(N,v_{UD})$ are significantly more similar to each other than to $(N,v_{IC})$. This is because, under \P-extendability, both of these are \P-extensions, while $(N,v_{IC})$ is in general not. A classical property, always satisfied by $(N,v_{IC})$, is \emph{monotonicity}, i.e.
\begin{equation}
    v_{IC}(S) \leq v_{IC}(T), \hspace{4ex} S \subseteq T \subseteq N.
\end{equation}
Later on, we denote the set of all monotonic $n$-player cooperative games as $\mathbb{M}^n$.
There exists a hierarchy of monotonicity concepts known as \emph{$k$-monotonicity} for $k \geq 1$, which contains monotonicity and positivity as special cases. Monotone games represent the broadest class in this hierarchy, referred to as \emph{1-monotone}, while positive games corresponds to the most restrictive class, known as \emph{$\infty$-monotone} or \emph{totally monotone}. For further details, see~\cite{Grabisch2016}.

\begin{proposition}\label{prop:special-games}
    Let $(N,\K,v)$ be \P-extendable intersection-closed incomplete game. Then $(N,v_R)$ and $(N,v_{UD})$ are its \P-extensions, while, in general, $(N,v_{IC})$ is only $\mathbb{M}^n$-extension.
\end{proposition}
\begin{proof}
    Positivity of $(N,v_{R})$ follows from the surpluses of this games, which can be rewritten for $S \in \K$ as $d_{v_{R}}(S) = \Delta_v(S)$. Positivity of $(N,v_{UD})$ follows from non-negativity of $\delta_v^\K(S)$ for every $S \subseteq N$, which is given by Definition~\ref{def:UD-value}. Monotonicity of $(N,v_{IC})$ follows from the fact that for \P-extendable $(N,\K,v)$,
    \begin{equation}\label{eq:5}
        v(S) \leq v(T), \hspace{4ex} S, T \in \K, S \subseteq T.
    \end{equation}
    Eq.~\eqref{eq:5} is satisfied from \P-extendability of $(N,\K,v)$, as there is $(N,w)$, a \P-extension, satisfying for every $S$, $T \in \K$, $S \subseteq T$,
    \begin{equation}
        v(S) = \sum_{X \subseteq S}d_v(X) \leq \sum_{X \subseteq T}d_v(X) \leq v(T).
    \end{equation}
    Now for every $S \subseteq T \subseteq N$, inequality $v_{IC}(S) \leq v_{IC}(T)$ corresponds to
    \begin{equation}\label{eq:6}
        v(c_\K(S)) \leq v(c_\K(T)), \hspace{4ex} c_\K(S),\ c_\K(T) \in \K.
    \end{equation}
    According to Lemma~\ref{lem:operator}, $c_\K(S) \subseteq c_\K(T)$, thus~\eqref{eq:6} follows from~\eqref{eq:5}.
\end{proof}
It is important to note that while results such as Proposition~\ref{prop:special-games} underscore differences between the underlying games that give rise to each value, these differences do not necessarily translate into disparities in the values themselves. This distinction arises because the Shapley value, viewed as a mapping from the space of games $\mathbb{R}^{2^n}$ to the space of payoff distributions $\mathbb{R}^n$, can map two very different games to surprisingly similar payoff allocations. In Section~\ref{sec:experiments}, we provide empirical results indicating that this phenomenon—that the R-value and the UD-value often lie close to each other, while both remain relatively distant from the IC-value—holds broadly in practice.

\subsection{Axiomatic comparison}
In Beál~\cite{Beal2020}, two axiomatizations of the IC-value were introduced. Both of these axiomatizations build heavily on the fact that $IC(v) = \phi(v_{IC})$. We modify these axiomatizations for the UD-value and sketch the modification for the R-value. Throughout this section, $\Gamma^n_\K$ denote the set of incomplete games on $n$ players and set system $\K$. Further we denote the set of complete games, which arise from the definition of the UD-values as
\begin{equation}
    \mathcal{UD}^n_\K = \left\{(N,w) \middle| d_w(T) = d_w(S) \text{ for } S,T \subseteq N, c_\K(S) = c_\K(T) \right\}.
\end{equation}

The first axiomatization builds on the classical result of Shapley~\cite{Shapley1953}, which characterizes the Shapley value using four axioms; \emph{Efficiency}, \emph{Additivity} \emph{Null player}, and \emph{Equal treatment of players}. To introduce the modifications of the axioms for incomplete games, assume $f^\K \colon \Gamma^n_\K \to \mathbb{R}^n$ where $\K$ is intersection-closed.

\textbf{Efficiency.} For every $(N, \K, v)$, it holds that
\[
\sum_{i \in N} f^\K_i(v) = v(N).
\]

\textbf{Additivity.} For every $(N, \K, v)$, $(N,\K,w)$ it holds that
\begin{equation}
f^\K(v + w) = f^\K(v) + f^\K(w).
\end{equation}

Recall a player $i \in N$ of a complete cooperative game $(N,w)$ is a \emph{null player}, if $w(S \cup i) = w(S)$ for every $S \subseteq N$. What follows is the axiom of the null player, which is a slightly weaker variant of the classical axiom, because it has to be satisfied only for games in $\mathcal{UD}^n_\K$.

\textbf{IC-Null player.} For every $(N, \K, v)$ with $v \in \mathcal{UD}^n_\K$, if $i \in N$ is a null player then
\begin{equation}
f^\K_i(v) = 0.
\end{equation}
The equal treatment axiom states that equal agents receive equal payoffs. Agents $i,j \in N$ are \emph{equal} in a complete game $(N,w)$ if $w(S \cup i) = w(S \cup j)$ for every $S \subseteq N \setminus \{i,j\}$. Once again, we provide a weaker variant restricted only to $\mathcal{UD}^n_\K$.

\textbf{IC-Equal treatment.} For every $(N, \K, v)$ with $v \in \mathcal{UD}^n_\K$, if $i,j \in N$ are equal, then
\begin{equation}
   f^\K_i(v) = f^\K_j(w). 
\end{equation}
The axioms listed above are insufficient to guarantee the uniqueness of the value. This limitation arises because it is mathematically possible for $f^\K(v) \neq f^\K(w)$ even when $v(S) = w(S)$ for every $S \in \K$, which is not desirable. We can fix this by the equality axiom.

\textbf{Equality.} For all $(N, \K, v)$, $(N,\K,w)$ satisfying $v(S) = w(S)$, $S \in \K$, it holds that
\begin{equation}
f^\K(v) = f^\K(w).
\end{equation}

In B\'{e}al et al.~\cite{Beal2020}, a different axiom was used, which stated that any game $(N,\K,v)$ with $v(S) = 0$ for every $S \in \K$ should satisfy for all $i,j \in N$,
\begin{equation}
    f^\K_i(v) = f^\K_j(v).
\end{equation}
It was proven that the Equality axiom can be derived from Efficiency, Additivity, and this alternative axiom. We find the Equality axiom more intuitive and easier to explain, however, we note the alternative can be used in our characterization as well.
\begin{proposition}\label{prop:beal-axiom}
    Let $\K$ be intersection-closed. The UD-value is the only function $f^\K \colon \Gamma^n_\K \to \R$ satisfying Efficiency, Additivity, IC-null player, IC-Equal treatment and Equality.
\end{proposition}
\begin{proof}
    (Existence) Since the UD-value is defined as the Shapley value of $(N,v_{UD})$, it is immediate that it satisfies the first four axioms. The equality follows from the fact that $w_{UD} = v_{UD}$ for any $(N,\K,v)$ and $(N,\K,w)$, which coincide on values of $S \in \K$.

    (Uniqueness) From the Equality axiom, $f^\K(v) = f^\K(v_{UD})$. Now we can follow the proof\footnote{In the original proof of Shapley, the unanimity games were used instead of games $(N,b_S)$.} of Shapley~\cite{Shapley1953} and rewrite $v_{UD}$ as a linear combination
    \begin{equation}\label{eq:7}
        v_{UD} = \sum_{S \in \K}\alpha_Sb_S,
    \end{equation}
    where $b_S$ is defined through its dividends as
    \begin{equation}
        d_{b_S}(T) = \begin{cases}
            \frac{1}{2^s} & T \subseteq S,\\
            0 & T \not\subseteq S.\\
        \end{cases}
    \end{equation}
    It is a straightforward exercise in elementary linear algebra to verify that games $(N, b_S)$ are linearly independent and are contained within \(\mathcal{UD}_\K^n\). Therefore, this set forms a basis for \(\mathcal{UD}_\K^n\), and any $v_{UD}$ can be uniquely expressed as a linear combination of these games. Using Additivity, we have $f^\K(v_{UD}) = \sum_{S \in \K} f^\K(\alpha_S b_S)$. The next observation is that for $i \in N \setminus S$, the player i is a null player, while for $j, k \in S$, these players are treated equally. By combining this observation with the axioms of Efficiency, IC-null player, and IC-equal treatment, we can derive the following:

\begin{equation}
f^\K_i(\alpha_S b_S) =
\begin{cases}
\frac{\alpha_S}{|S|} &  i \in S,\\
0 &  i \notin S.
\end{cases}
\end{equation}
\end{proof}

The only difference between our axiomatization and the one in B\'{e}al et al.~\cite{Beal2020} is in the IC-Null player and IC-Equal treatment axioms; namely, restriction $v \in \mathcal{UD}_\K^n$ should be substituted with $v \in \mathcal{IC}_\K^n$ where $\mathcal{IC}_\K^n$ is the set of all complete games which arise as $(N,v_{IC})$. Similarly, for the R-value, one would have to resort to $v \in \mathcal{R}_\K^n$ representing the set of all complete games $(N,v_R)$. Following the proof of Proposition~\ref{prop:beal-axiom}, the only difference occurs at Eq.~\eqref{eq:7} where bases of $\mathcal{IC}_\K^n$ and $\mathcal{R}_\K^n$ have to be considered. For $\mathcal{IC}_\K^n$, the basis of \emph{lower games} was considered (see B\'{e}al et al.~\cite{Beal2020} for details), while for $\mathcal{R}_\K^n$, one can consider the basis of \emph{unanimity games} $(N,u_S)$ for $S \in \K$. 

The second axiomatization of the IC-values employs two additional axioms: \emph{$\phi$-consistency} and \emph{Invariance from irrelevant changes}. Both of these axioms assume that $\K$ is not fixed but rather an argument of $f$ as well. Therefore, formally, instead of considering the allocation rule of form $f^\K \colon \Gamma^n_\K \to \mathbb{R}$, we consider $f \colon \Gamma_{\mathcal{UD}}^n \to \mathbb{R}$, where 
\[\Gamma^n_{\mathcal{UD}} = \bigcup_{\K \subseteq 2^N ...\text{intersection-closed}}\Gamma^n_{\K}.\]

\textbf{$\phi$-consistency.} For an incomplete game $(N,\K,v)$ where $\K = 2^N$, it holds
\begin{equation}
f^{\K}(v) = \phi(v).
\end{equation}

Value $f$ , satisfying the $\phi$-consistency axiom, is a generalization of the Shapley value to intersection-closed set systems. The Invariance from irrelevant changes axiom can be viewed as generalization of the Equality axiom to two, possibly different set systems.

\textbf{Invariance from irrelevant changes.} For any $(N,\K_1,v)$ and $(N,\K_2,w)$, if $v_{UD} = w_{UD}$ then
\begin{equation}
f^{\K_1}(v) = f^{\K_2}(w).
\end{equation}

For $\K_1 = \K_2$, $v_{UD} = w_{UD}$ if and only if $v(S) = w(S)$ for every $S \in \K_1$. Therefore, the Equality axiom follows from Invariance from irrelevant changes.

\begin{proposition}\label{prop:char2}
    The UD-value is the only function $f \colon \Gamma_{\mathcal{UD}}^n \to \mathbb{R}$ satisfying $\phi$-consistency, and Invariance from irrelevant changes.
\end{proposition}
\begin{proof}
(Uniqueness) Consider intersection-closed $(N,\K,v)$. As Invariance from irrelevant changes implies Equality, we have $f^\K(v) = f^\K(v_{UD})$. The key step is in considering Invariance from irrelevant changes for $\K_1=\K$ and $\K_2 = 2^N$. It follows
\begin{equation}\label{eq:8}
    f^\K(v_{UD}) = f^{2^N}(v_{UD}).
\end{equation}
Finally from $\phi$-consistency, it follows $f^{2^N}(v_{UD}) = \phi(v_{UD}) = UD^\K(v)$.

\textbf{(Existence)} From Definition~\ref{def:UD-value}, it is straightforward that the UD-value satisfies $\phi$-consistency. The invariance under irrelevant changes also follows immediately, as $UD^{\K_1}(v) = \phi(v_{UD}) = \phi(w_{UD}) = UD^{\K_2}(w)$, given that $v_{UD} = w_{UD}$.

\end{proof}
Once again, characterization from Proposition~\ref{prop:char2} can be easily modified for the R-value, by modifying the Invariance from irrelevant changes axiom. Eq.~\eqref{eq:8} should be satisfied for games $(N,\K_1,v)$ and $(N,\K_2,w)$ satisfying $v_{R} = w_{R}$.

Both presented characterizations show similarities of the three values. We employ axioms from the existing characterization of the R-value to show more of the dissimilarities. One characterization of the R-value employs the axiom of \emph{Fairness}. 

\textbf{Fairness.} For every $i,j \in S$ and $S \in \K$,
\begin{equation}
    f_i^\K(v) - f_i^{\K \setminus S}(v) = f_j^\K(v) - f_j^{\K \setminus S}(v).
\end{equation}
The fairness axiom states that the loss of information regarding value of $S$ should affect all agents in $S$ in the same way. This behavior is typical for the R-value, however, the following example shows that the UD-value does not hold this property.
\begin{example}\label{ex:1}
    Let $N = \{1,2,3\}$ and $\K = \{\emptyset, \{1\}, \{1,2\}, \{1,2,3\}\}$. Further, let $v(N)=1$ and $v(S) = 0$, otherwise. If we consider $S = \{1,2\}$, the fairness is not satisfied for the UD-value. This is because after $\{1,2\}$ is taken away, surplus of all coalitions except for subsets of $\{1,2\}$ is uniformly decreased, while the surplus of $\{1,2\}$ and $\{2\}$ is equally increased and the surplus of $\{1\}$ remains the same. While the decrease of the surpluses affect the value of both of the players the same, the increase favors player $2$.
\end{example}
Another axiom, which occurs in a different characterization of the R-value is the axiom of \emph{Balanced contributions}. This axiom states that for any two players, the amount that each player would gain or lose by the other player's withdrawal from the game should be equal. Formally, a withdrawal of an agent $i$ from an incomplete game $(N,\K,v)$ can be captured by $(N \setminus i, \K_{-i},v_{-i})$ where $\K_{-i} = \left\{S \in \K \middle| i \notin S\right\}$ and $v_{-i}$ is the restriction of $v$ to $N \setminus i$.

\textbf{Balanced contributions.}
For every $i,j \in N$, it holds
\begin{equation}
    f^\K_i(v) - f^{\K_{-j}}_i(v_{-j}) = f^{\K}_j(v) - f^{\K_{-i}}_j(,v_{-i}).
\end{equation}
Even if $\K_{-i}$ and $\K_{-j}$ are intersection closed, which does not hold for every $\K$, the balanced contributions axiom does not have to be satisfied for the UD-value.
\begin{example}\label{ex:2}
    Let $N = \{1,2,3\}$ and $\K = \{\emptyset, \{1\}, \{2\}, \{1,2\}, \{2,3\}, \{1,2,3\}\}$. Further, let the values be given as
    \begin{itemize}
        \item $v(\emptyset) = v(\{2\}) = 0$,
        \item $v(\{1\}) = 1$,
        \item $v(\{1,2\}) = v(\{2,3\}) = 2$, and
        \item $v(\{1,2,3\}) = 4$.
    \end{itemize}
    While $UD_1^\K(v) = UD_3^\K(v) = 1.5$, we have $UD_1^{\K_{-3}}(v_{-3}) = 1.5$ and at the same time $UD_3^{\K_{-1}}(v_{-1}) = 1$.
\end{example}

The third and last is the axiom of \emph{Symmetric Partnership}, which occured in axiomatization by Albizzuri et al.~\cite{Albizuri2022}. This axiom employs the so called \emph{coalition of partners}, which for an incomplete game $(N,\K,v)$ is defined as a coalition $P \subseteq N$ satisfying for every $S \in \K$ for which $P \setminus S \neq \emptyset$ that
\begin{enumerate}
 \item $S \setminus P \in \K \implies v(S) = v(S \setminus P)$,
 \item $S \setminus P \notin \K \implies \forall T \in \K, T \subseteq S: v(T) = 0$.
\end{enumerate}
Coalition $P$ is a coalition of partners if no proper subset of players from $P$ makes any contribution to any coalition outside $P$. For this partnership, the agents are rewarded equally.

\textbf{Symmetric Partnership.}
If $P$ is a coalition of partners in $(N,\K,v)$, then for every $i,j \in P$,
\begin{equation}
    f^\K_i(v) = f^\K_j(v).
\end{equation}
The Symmetric Partnership axiom is not satisfied for the UD-value. This can be shown on a modification of Example~\ref{ex:2}.
\begin{example}
    Let $N = \{1,2,3\}$ and $\K = \{\emptyset,\{1\},\{2\},\{1,2\},\{2,3\},\{1,2,3\}\}$. Further, let the values be given as
    \begin{itemize}
        \item $v(\emptyset) = v(\{1\}) = v(\{2\}) = v(\{1,2\}) = 0$, and
        \item $v(\{2,3\}) = v(\{1,2,3\}) = 1$.
    \end{itemize}
    Coalition $P = \{1,2\}$ is a coalition of partners, however, $UD^\K_1 = 0$, while $UD^\K_2 = 0.25$.
\end{example}

We note that the three axioms cannot be satisfied for the IC-value as well as for the UD-value, however, we omit the examples for this value.

\subsection{Experimental comparison}\label{sec:experiments}
We conducted several computational experiments on intersection-closed systems to compare the R-value, the UD-value, and the IC-value. Two main observations emerged from these experiments:

\begin{enumerate}
    \item On average, the UD-value and the R-value lie much closer to each other than either does to the IC-value.
    \item Among the three values, the IC-value most uniformly distributes the grand coalition’s payoff among the agents.
\end{enumerate}

\subsubsection{Difference of values}
In our first experiment, we compare the $\ell_1$-norms of the differences between the values for a given intersection-closed game $(N, \mathcal{K}, v)$. Specifically, we compute $||R(v) - IC(v)||_1$, $||R(v) - UD(v)||_1$, and $||UD(v) - IC(v)||_1$. To capture the average behavior of these norms for a given $\mathcal{K}$, we generate multiple random games where the value $v(S)$ for each $S \in \mathcal{K}$ is selected uniformly from interval $[0, 1]$. We then compute the average and the standard deviation of all the norms obtained.

We encode each coalition $S \subseteq N$ as a binary number of length $n$. In this encoding, the bits corresponding to elements in $S$ (counting from the least significant bit) are set to $1$, while the rest are set to $0$. For example, if $n = 5$ and $S = \{1, 2, 5\}$, we set bits at positions $1, 2,$ and $5$ to $1$. This results in the binary number $10011$, which equals $19$ in decimal form.

Similarly, to uniquely represent a set system $\K \subseteq 2^N$, we use a binary number of length at most $2^n$. In this number, the bit at position $i$ corresponds to the presence of the coalition $S \in \mathcal{K}$ that has the integer representation $i$. For $\K = \{\emptyset, \{1\}, \{1,2\}, \{1,2,3\}\}$, the bit number corresponds to $10000111$, which is $135$ in decimal form.

Figure~\ref{fig:3players-difference} illustrates the results for 3 players and all intersection-closed set systems. Each set system is encoded into an integer, which is assigned to the x-axis, while the y-axis shows the average differences of the values and the standard deviance.

\begin{figure}[ht]
\centering
\includegraphics[width=\linewidth]{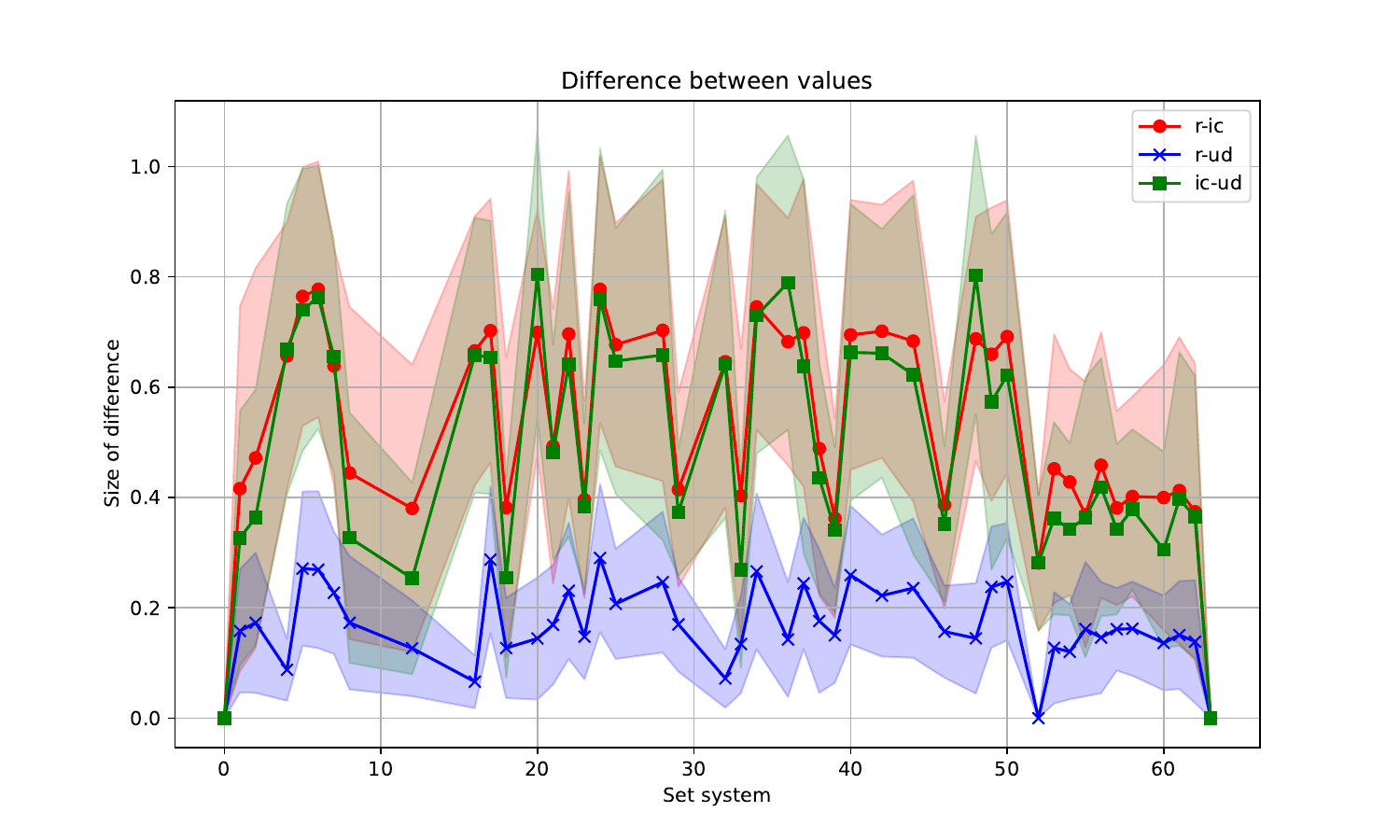}
\caption{Average $\ell_1$-norms of differences between the R-value, UD-value, and IC-value for all intersection-closed set systems with $n = 3$ players. Each set system is represented on the x-axis by its integer encoding, and the y-axis shows the average norm of the differences between the values, computed over 100 randomly generated games games with values selected uniformly from $[0,1]$.}
\label{fig:3players-difference}
\end{figure}

From Figure~\ref{fig:3players-difference}, we observe that the difference between the R-value and the UD-value is, on average, the smallest among the three pairs of values. This trend holds across all intersection-closed set systems when $n = 3$. However, extending this experiment to $n > 3$ becomes considerably more challenging due to the exponential increase in the number of intersection-closed set systems. Even for $n = 4$, there are $2{,}271$ distinct intersection-closed set systems containing both $\emptyset$ and $N$. Although we can still run the experiment in principle, the sheer volume of data makes it difficult to present a visualization analogous to Figure~\ref{fig:3players-difference}. For larger values of $n$, any attempt at exhaustive enumeration quickly becomes infeasible in practice.

To address these issues, we modify our approach in Figure~\ref{fig:diff_order}. Instead of plotting individual average differences for each set system, we focus on the relative ordering of the differences among the three value pairs. For each set system, we compute the average differences $||R(v) - IC(v)||_1$, $||R(v) - UD(v)||_1$, and $||UD(v) - IC(v)||_1$. This average is computed over 100 randomly generated games with values from intervals $\left[0,1\right]$. We then record the order of these differences—identifying which is the smallest, which is the second largest, and which is the largest. After processing all set systems, we aggregate the results and plot the frequencies of each difference being the smallest, second largest, or largest.

For $n = 3$ and $n = 4$, we perform this analysis exhaustively by considering all possible intersection-closed set systems. However, for $n > 4$, due to computational constraints, we sample random intersection-closed set systems and conduct the experiment on this sample.

To determine the size of the representative sample, we use the following statistical method~\cite{Cochran1977}. First, we sample 30 random intersection-closed set systems to estimate the standard deviation of the average differences, denoted by $s$. The required sample size $s_n$ is then computed using the formula
\[
s_n = \left( \frac{Z \cdot s}{\text{E}} \right)^2,
\]
where $Z$ is the z-score corresponding to the desired confidence level, $s$ is the standard deviation estimated from the pilot sample, and $\text{E}$ is the acceptable margin of error. For both $n=5$ and $n=6$, we set $Z=1.96$ (corresponding to 95\% confidence level) and $E = 0.01$ and received $s_n$ between 3,000 -- 4,000 samples.

\begin{figure}[ht!]
    \centering
    \includegraphics[height=0.19\textheight,width=0.49\linewidth]{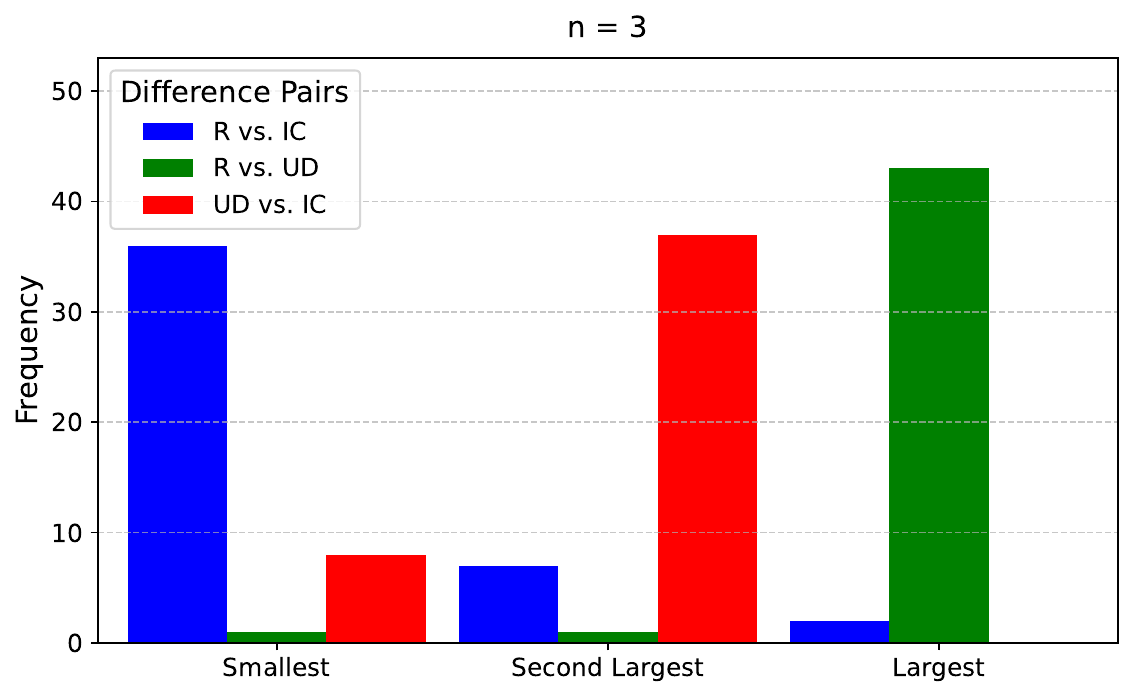}
    \includegraphics[height=0.19\textheight,width=0.49\linewidth]{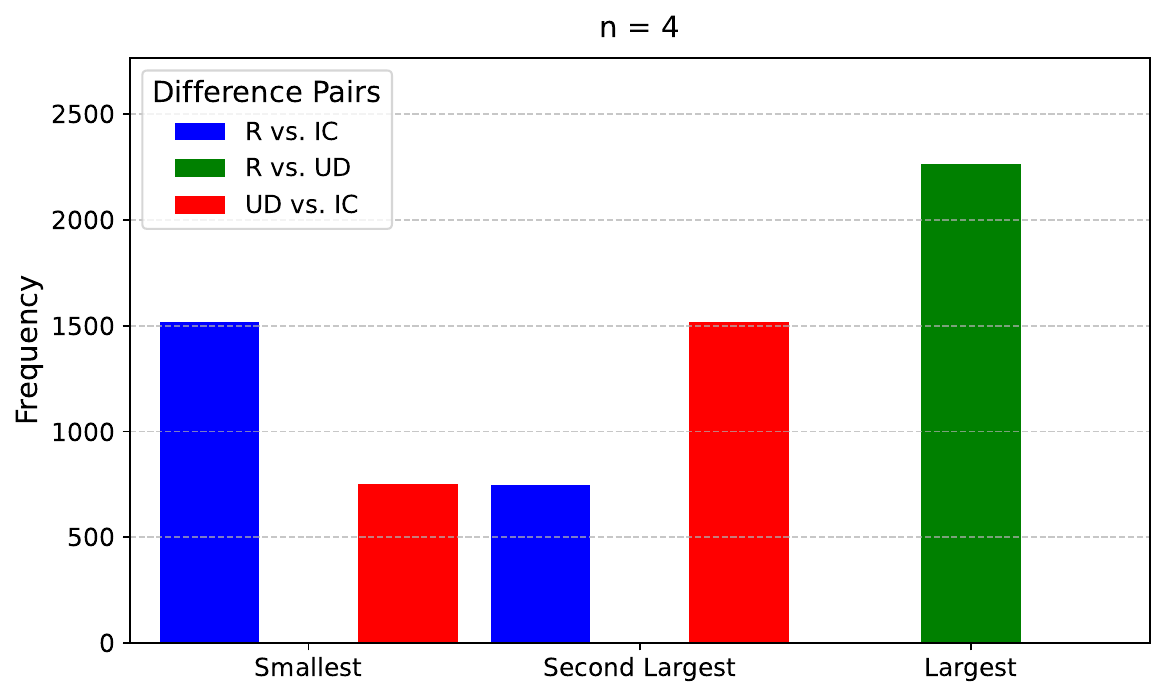}
    \includegraphics[height=0.19\textheight,width=0.49\linewidth]{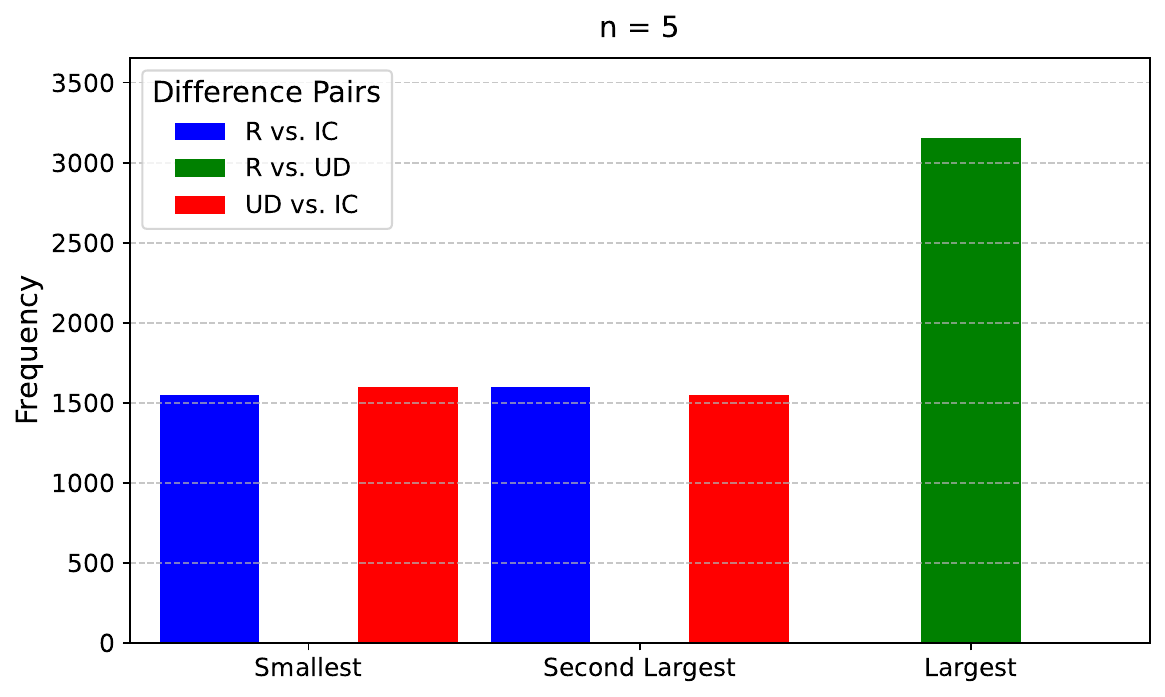}
    \includegraphics[height=0.19\textheight,width=0.49\linewidth]{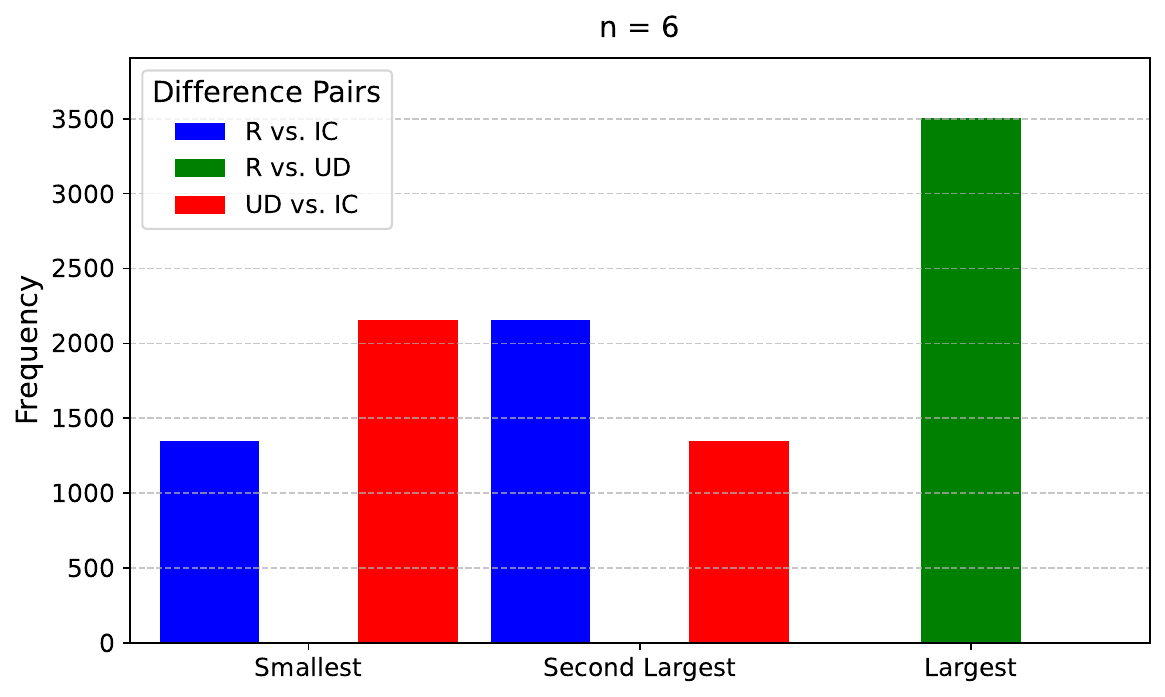}
    \caption{Frequencies of each $\ell_1$-norm difference being the smallest, second largest, or largest for all intersection-closed set systems with $n = 3, 4, 5, 6$ players. The y-axis shows the frequency of each rank, aggregated over sampled or exhaustively evaluated systems, highlighting consistent trends as $n$ increases.}
    \label{fig:diff_order}
\end{figure}

\begin{figure}[ht!]
\centering
\includegraphics[width=0.49\linewidth]{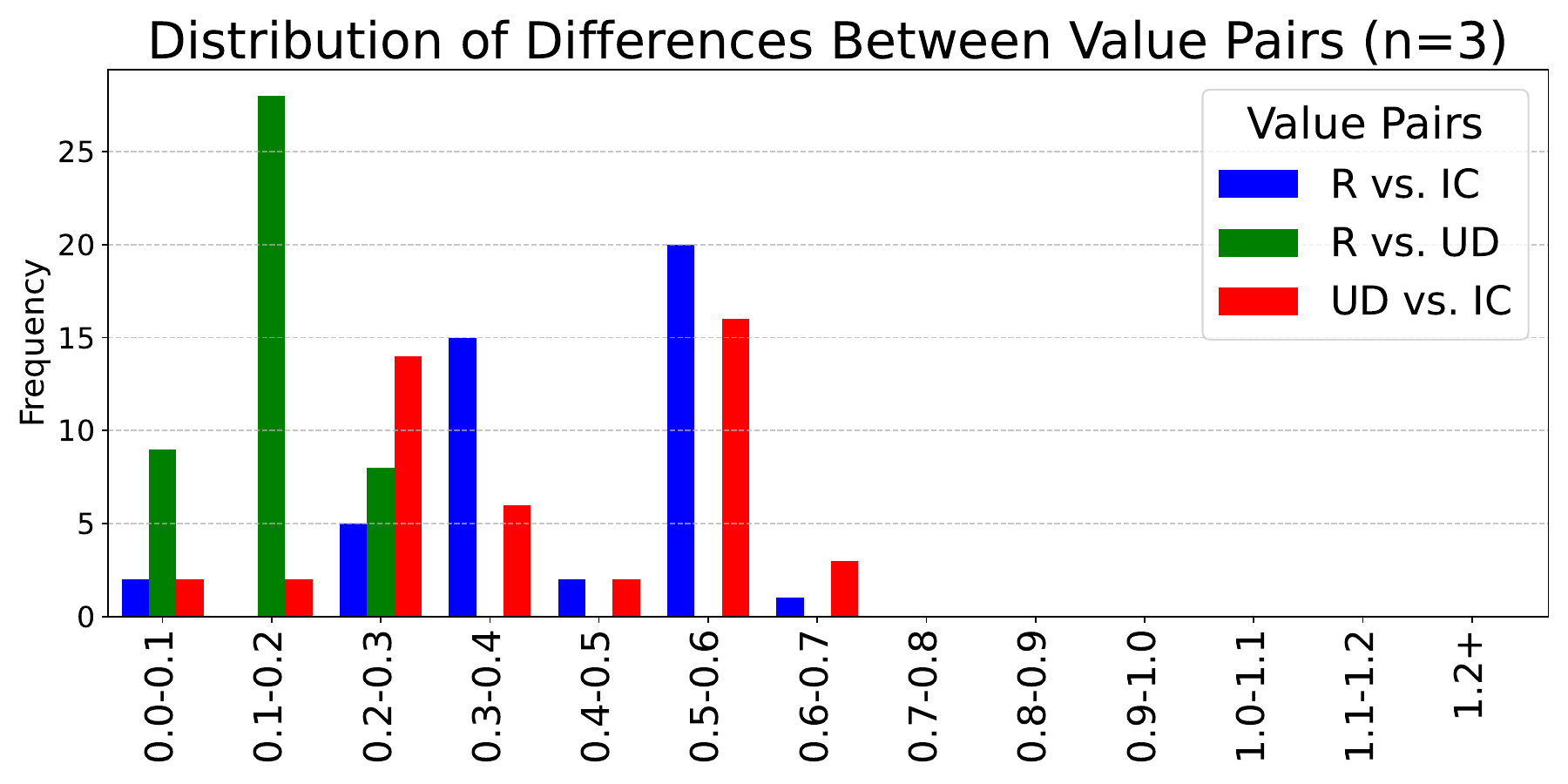}
\includegraphics[width=0.49\linewidth]{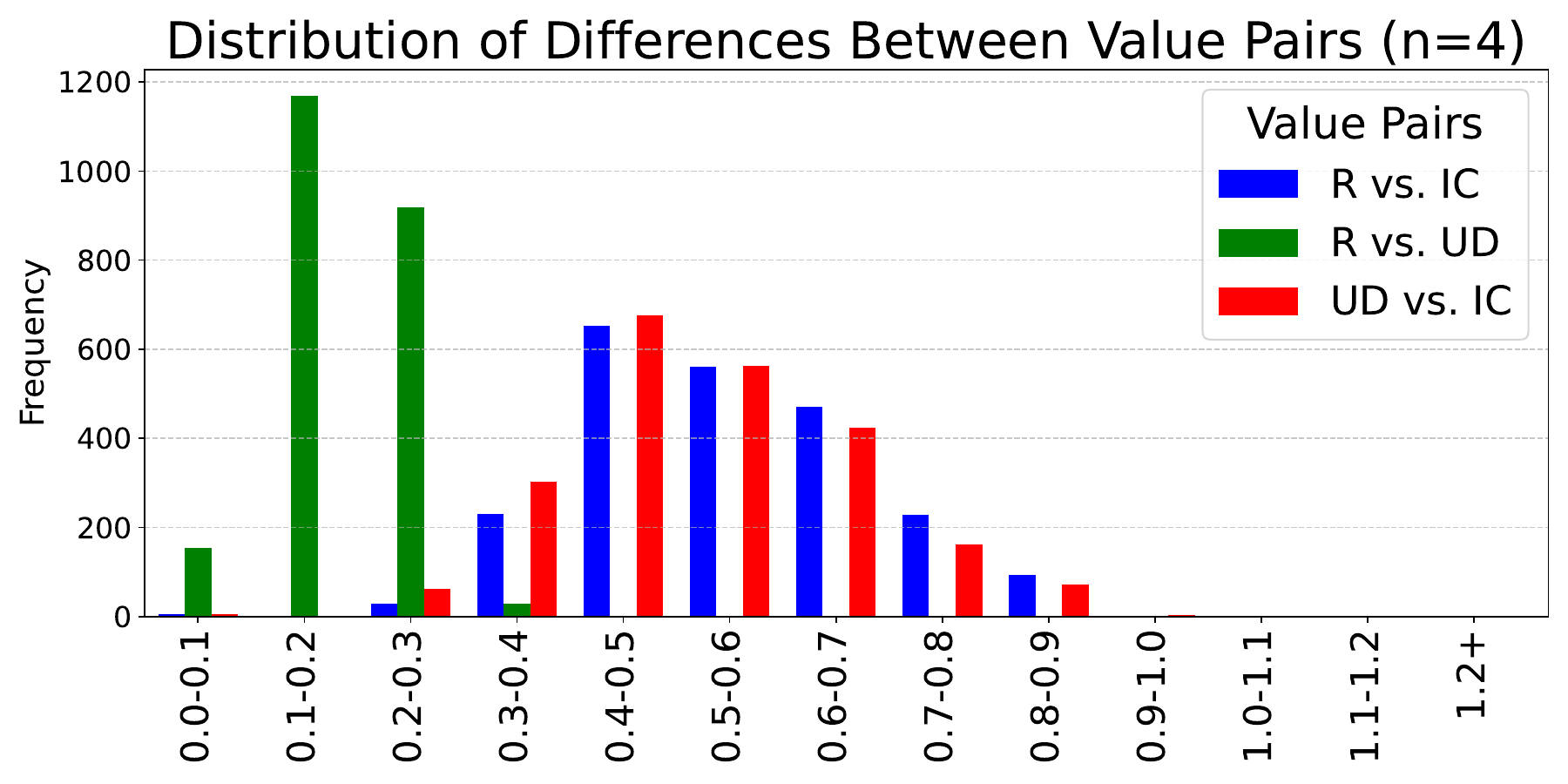}
\includegraphics[width=0.49\linewidth]{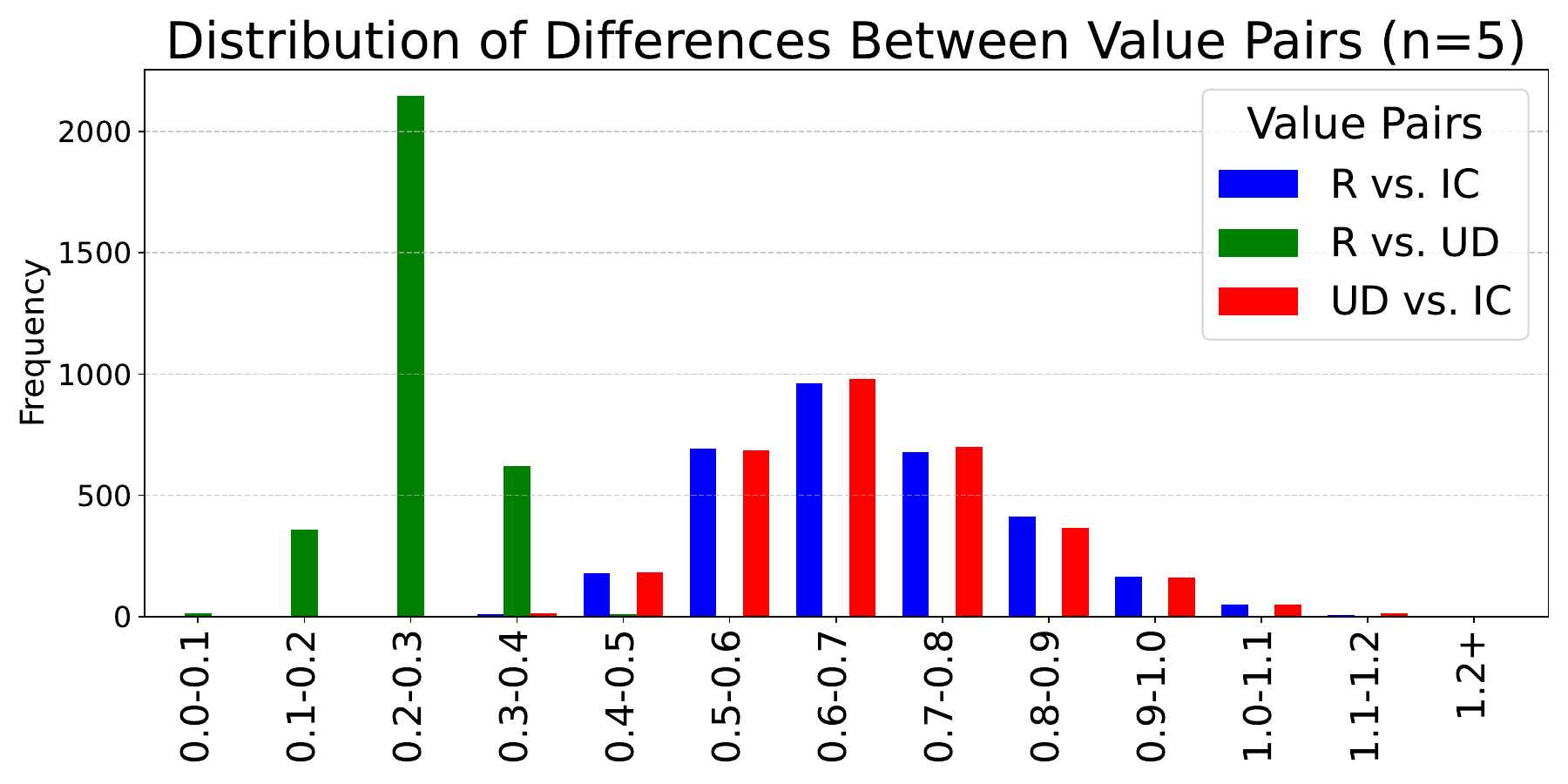}
\includegraphics[width=0.49\linewidth]{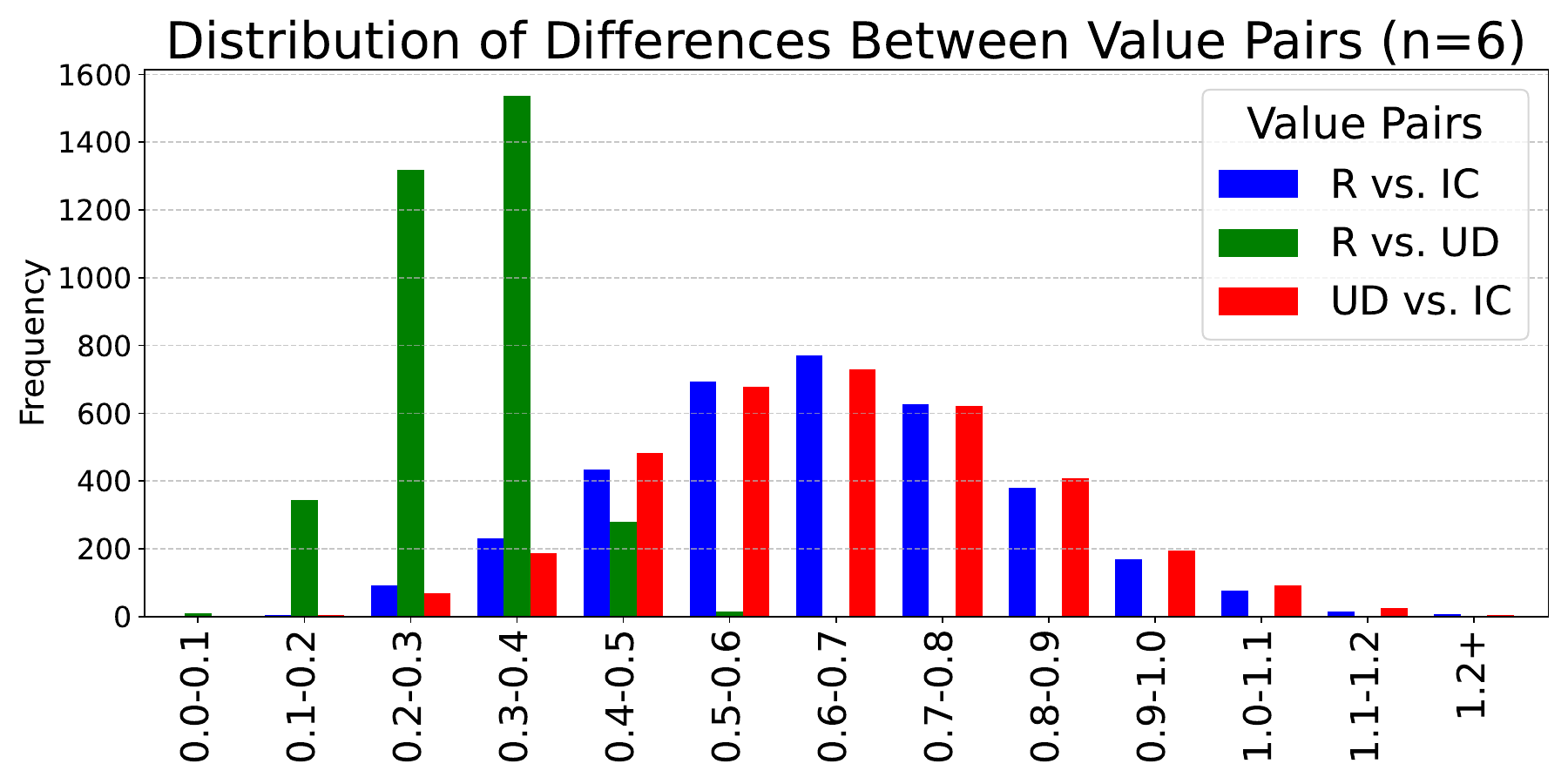}
\caption{Histograms of the average differences between value pairs for $n = 3, 4, 5, 6$. Each histogram shows the frequency distribution of average differences divided into intervals of width 0.1 from 0 to 1.2.}
\label{fig:diff_range}
\end{figure}

In Figure~\ref{fig:diff_order}, we observe that the behavior is similar for all values of $n$. Specifically, the difference between the R-value and the UD-value consistently remains the smallest, while the difference between the R-value and the IC-value is the largest. As $n$ increases, the frequency of the R-value and IC-value difference being the largest decreases, whereas the frequency of the UD-value and IC-value difference being the largest increases. Additionally, from Figure~\ref{fig:3players-difference} for $n = 3$, we notice that these differences follow very similar distributions in terms of average value and variance. Therefore, Figure~\ref{fig:diff_order} may not be particularly descriptive when comparing these two differences.

To compare these two differences more quantitatively, we provide a third plot, Figure~\ref{fig:diff_range}. This figure represents the frequencies of average differences divided into categories ranging from 0 to 1.2, with each category having a width of 0.1. Interestingly, these distributions resemble a normal distribution.

\begin{figure}[ht!]
    \centering
    \includegraphics[height=0.35\textheight,width=\linewidth]{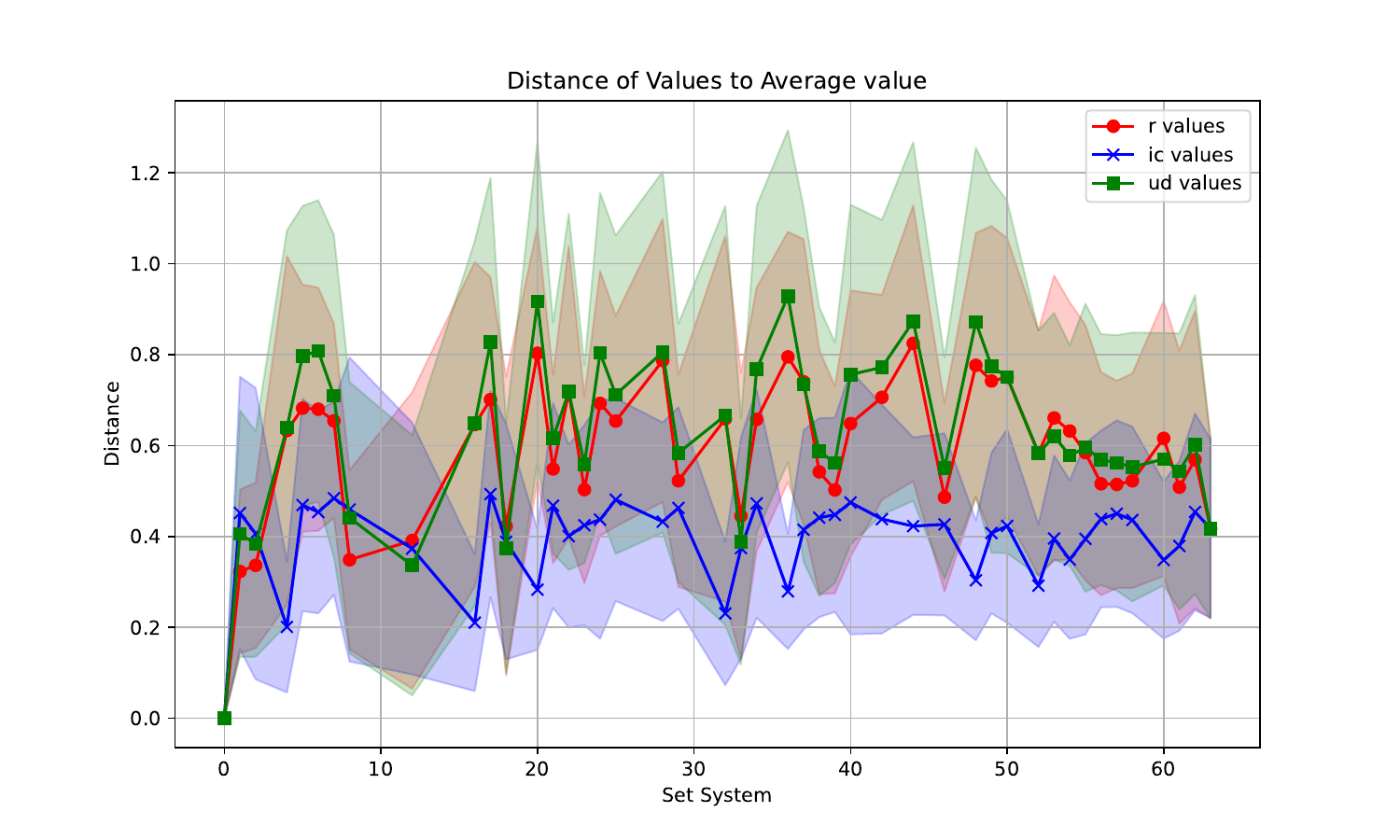}
    \caption{Average $\ell_1$-norms of distances between the R-value, UD-value, IC-value and the Equal Division Rule for 3 players. Each set system is represented on the x-axis by its integer encoding, while the y-axis shows the average distance between the values and ED, computed over 100 randomly generated games with values selected uniformly from $[0,1]$.}
    \label{fig:dist_3}
\end{figure}

\subsubsection{Distance from the Equal Division Rule}
In this experiment, we examine how evenly $v(N)$ is distributed among the players the different values. Formally, we approach this experiment similarly to the previous one, computing the differences between the three values and the so-called \emph{equal division rule} (ED), defined as $ED_i(v) = \frac{v(N)}{n}$ for every $i \in N$.

Figure~\ref{fig:dist_3} illustrates the case for three players. Notably, when we look at incomplete games, the \emph{UD-value} shows the largest deviation from the ED rule. Moreover, the distances of the \emph{R-value} and the \emph{UD-value} from the ED rule appear similarly distributed across the set systems, which is expected given how close these two values often are to each other.

For larger $n$, we conducted a similar experiment to that in Figure~\ref{fig:diff_order}, counting how frequently each value (averaged over all players) is (1) the closest, (2) the second closest, or (3) the furthest from the ED rule. Across $n = 3, 4, 5, 6$, we observe consistent behavior: most of the time, the \emph{IC-value} is the closest to ED, and the \emph{UD-value} is the furthest. These findings are presented in Figure~\ref{fig:diff_order}.

However, a qualitative experiment, similar to the one shown in Figure~\ref{fig:dist_order}, reveals that the results from Figure~\ref{fig:diff_order} can be misleading . Although the average distances of the R-value and the UD-value differ, their distributions exhibit similar patterns, again resembling a normal distribution.

\begin{figure}[ht!]
    \centering
    \includegraphics[height=0.19\textheight,width=0.49\linewidth]{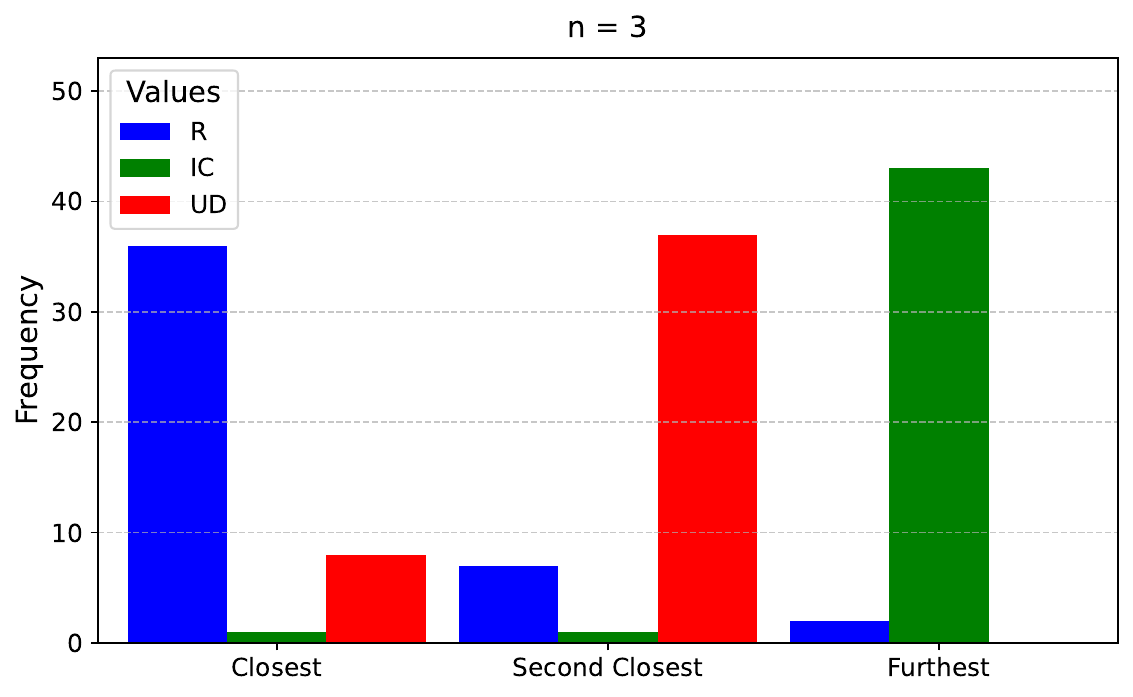}
    \includegraphics[height=0.19\textheight,width=0.49\linewidth]{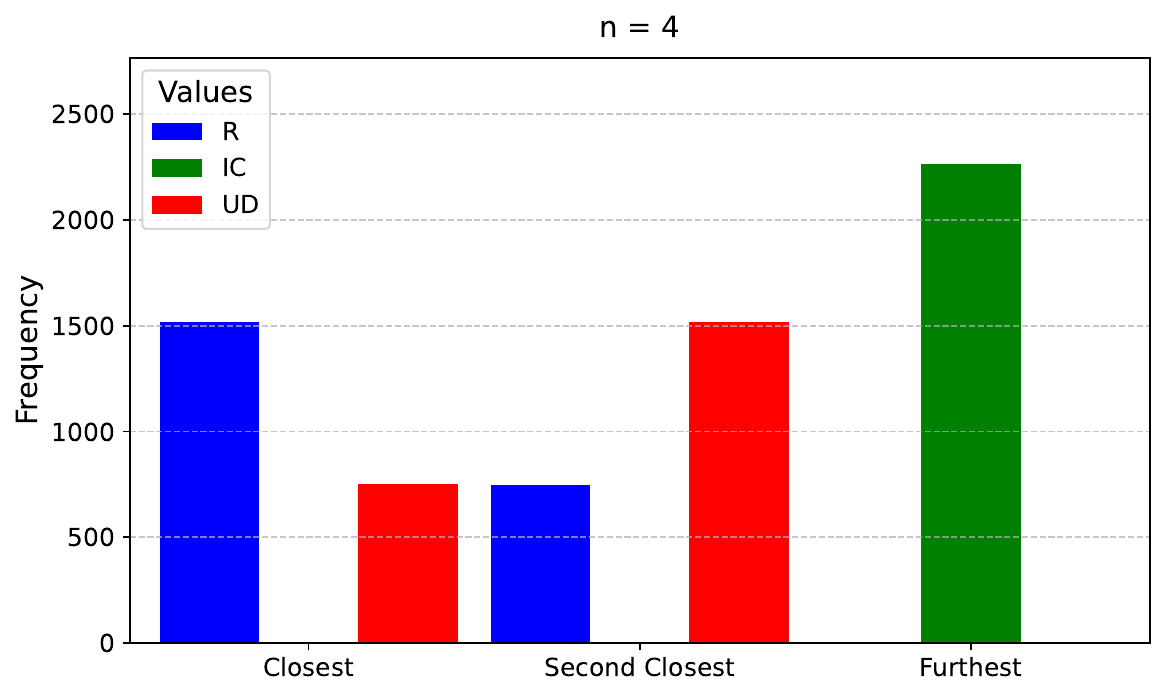}
    \includegraphics[height=0.19\textheight,width=0.49\linewidth]{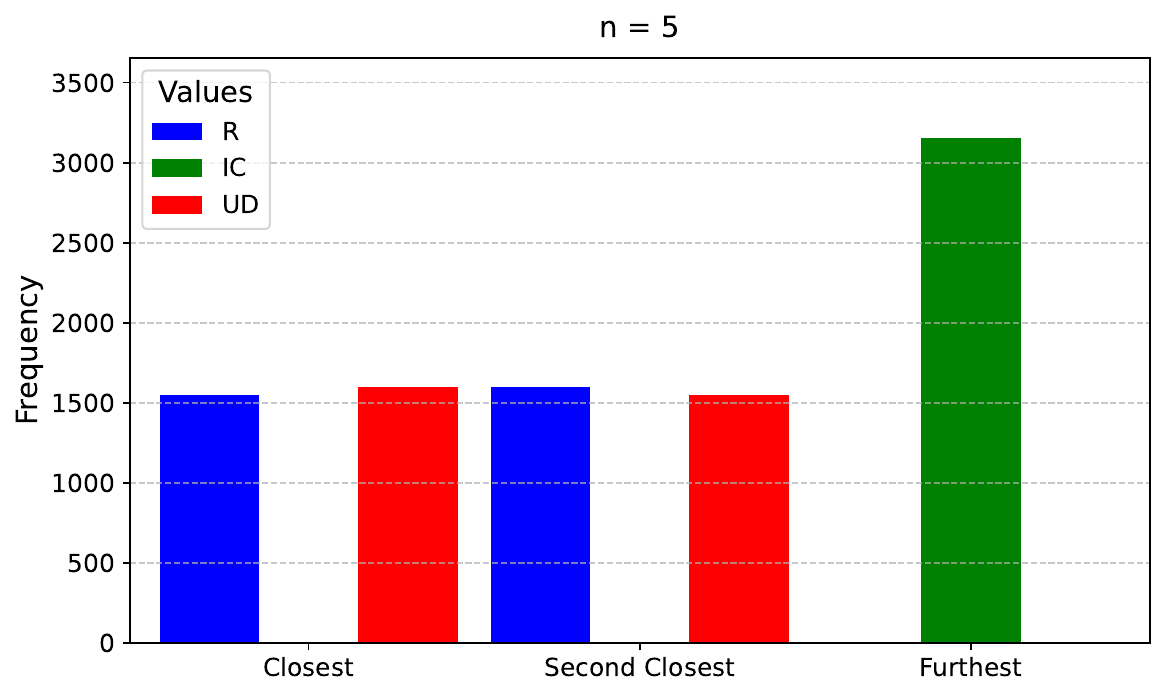}
    \includegraphics[height=0.19\textheight,width=0.49\linewidth]{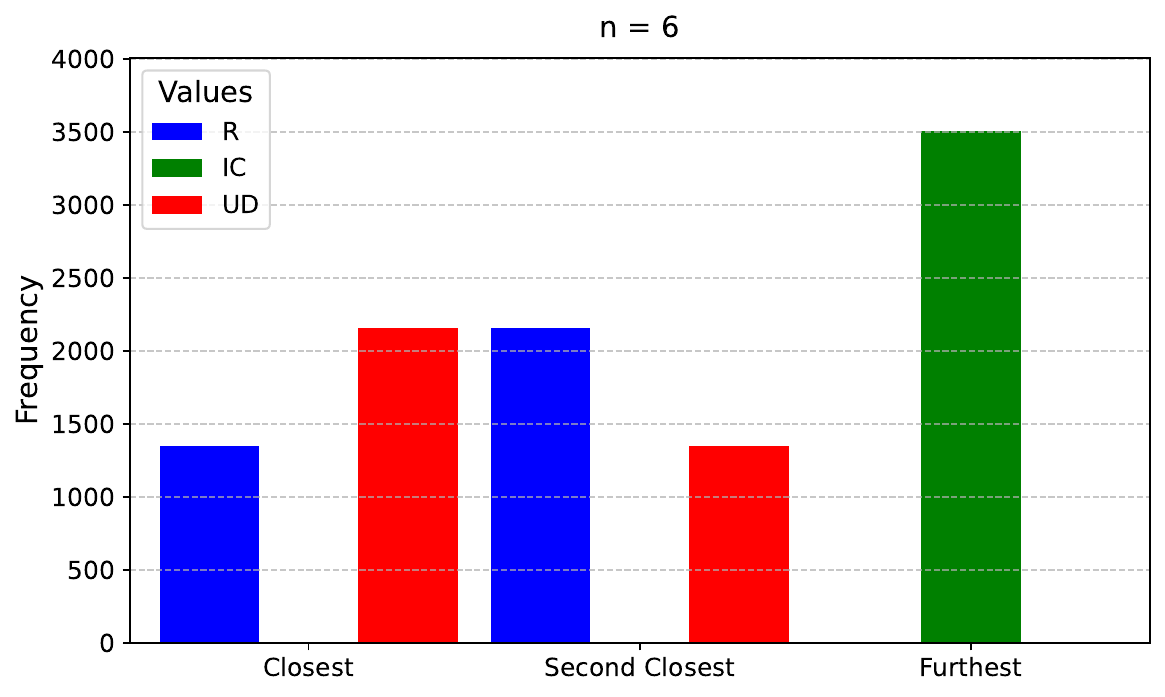}
    \caption{Frequencies of each value’s distance from the ED rule being the closest, second closest, and largest for $n = 3, 4, 5, 6$. The y-axis shows the frequency of each rank, aggregated over sampled or exhaustively evaluated systems, higlighting consistent trends as $n$ increases.}
    \label{fig:dist_order}
\end{figure}

\begin{figure}[ht!]
    \centering
    \includegraphics[width=0.48\linewidth]{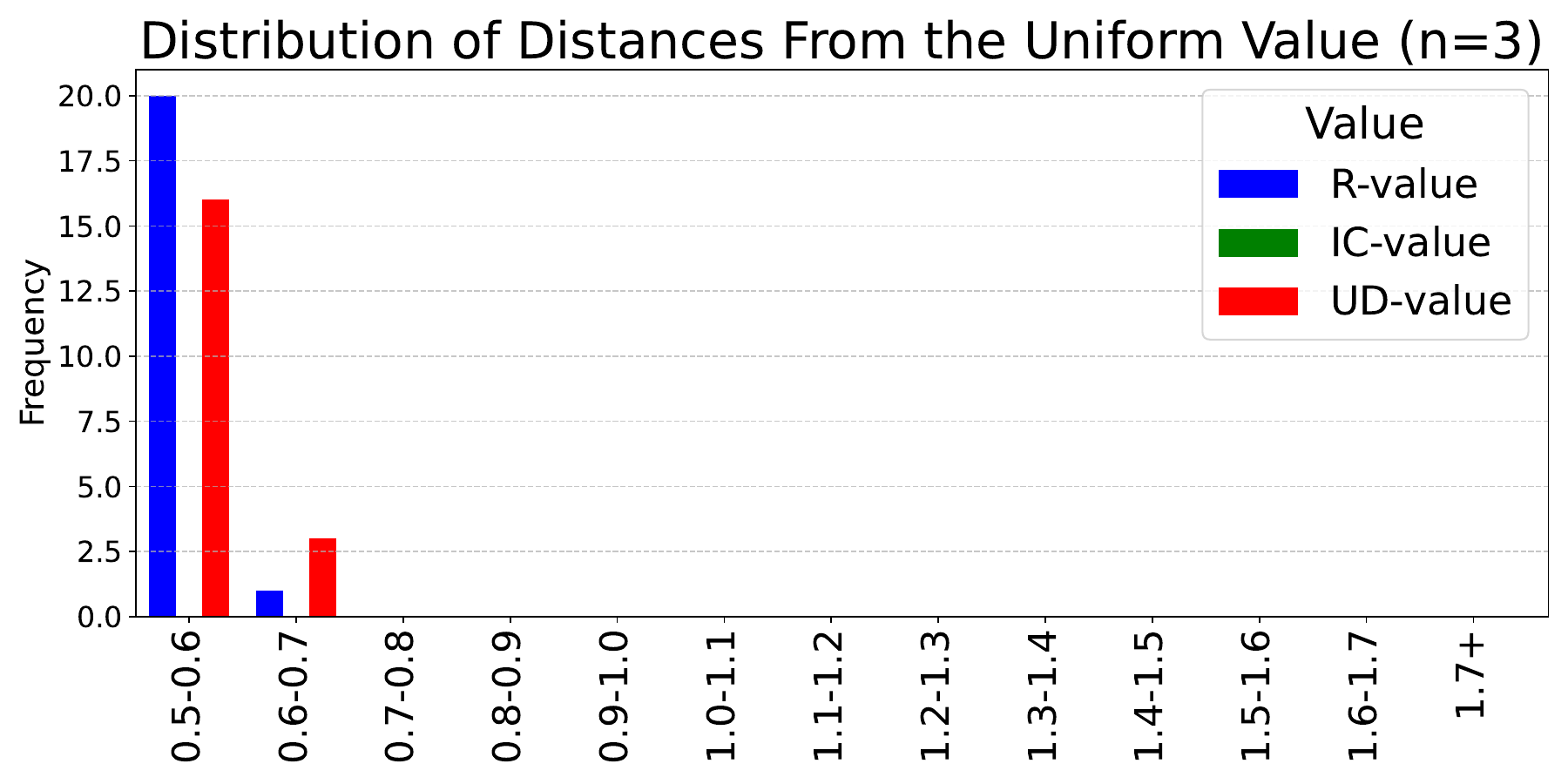}
    \includegraphics[width=0.48\linewidth]{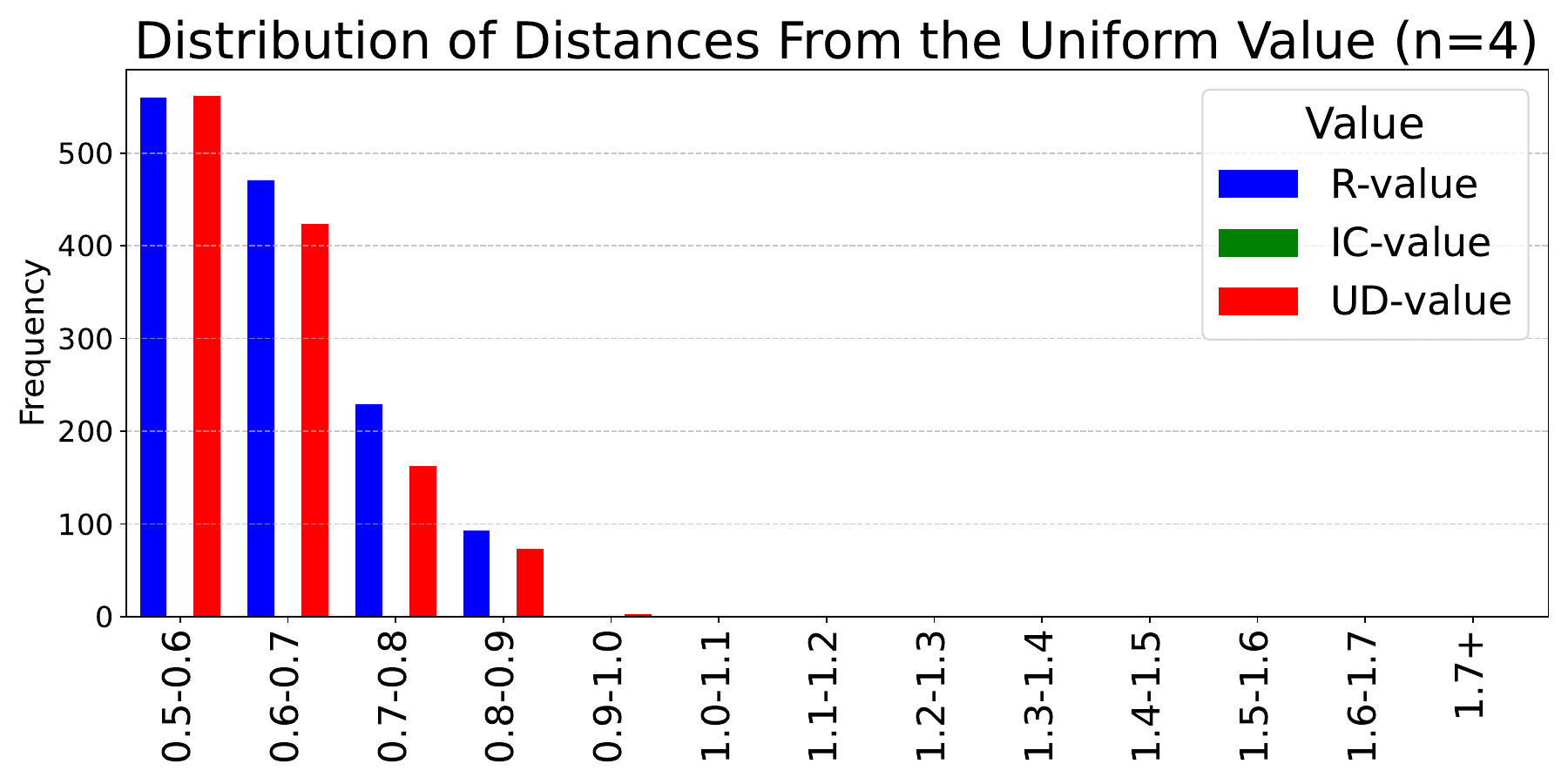}
    \includegraphics[width=0.48\linewidth]{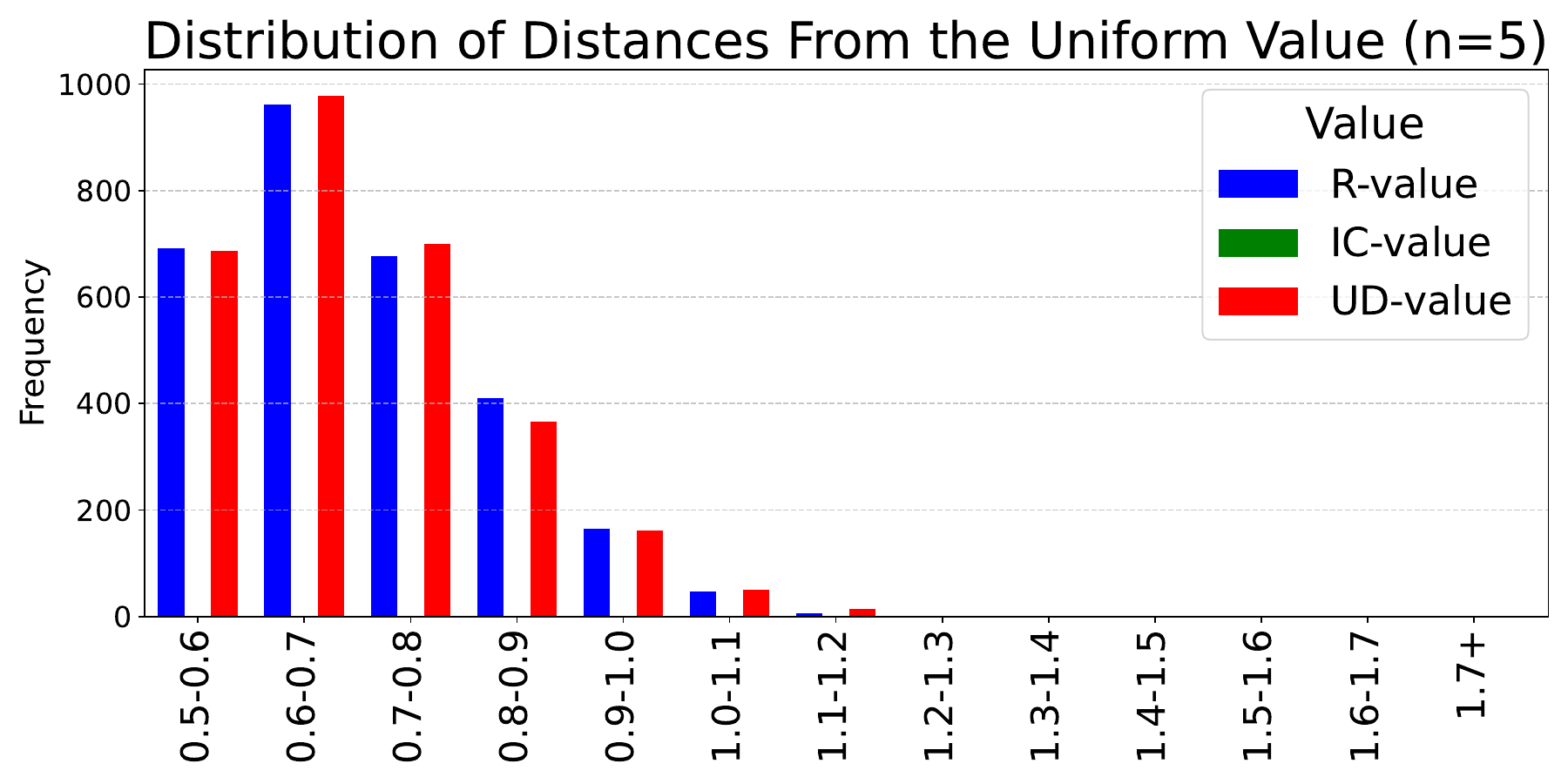}
    \includegraphics[width=0.48\linewidth]{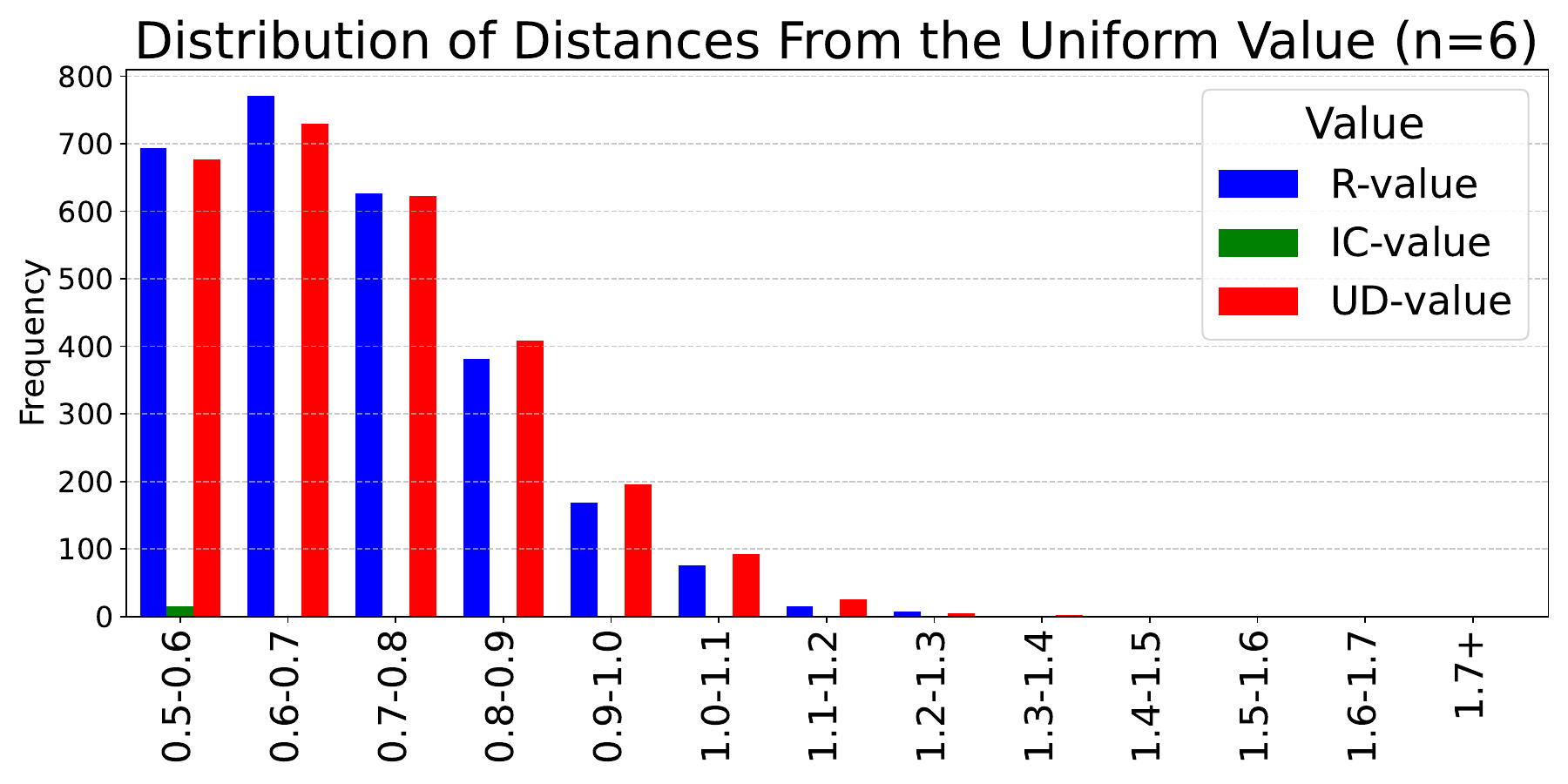}
    \caption{Histograms of the average distances from the ED rule for $n = 3, 4, 5, 6$. Each histogram shows the distribution of these distances (ranging from 0.5 to 1.7) in bins of width 0.1. The charts reveal how frequently the R-value, UD-value, and IC-value differ from an equal split, across various intersection-closed set systems.}
    \label{fig:dist_range}
\end{figure}

\section{Conclusions}
We introduced the \emph{uniform-dividend} (UD) value for incomplete cooperative games and proved that it is uniquely determined when the set of known coalitions is intersection-closed. Our comparisons with the R-value and the IC-value highlight how the UD-value complements these existing allocation rules, offering an alternative in intersection-closed system. Numerical experiments further indicate that the UD-value and the R-value consistently lie closer to each other than either does to the IC-value.

A key strength of the UD-value is its natural interpretation as the \emph{expected Shapley value over all positive extensions} of the incomplete game. In situations where only some coalitions’ values are known, this viewpoint shows that the UD-value fairly reflects each agent’s average contribution across all admissible extensions of the game. 

As the number of players grows, the proportion of intersection-closed set systems becomes vanishingly small—yet, interestingly, our experiments show that the share of systems for which the UD-value is unique keeps growing, ultimately forming the majority among all set systems. Investigating this phenomenon in a broader scope to understand exactly when such uniqueness arises is a direction for future work.

\section*{Acknowledgements}
The project was supported by the Czech Science Foundation grant no. 25-15714S. The author was further supported by the Charles University Grant Agency (GAUK 206523) and the Charles University project UNCE 24/SCI/008.

\appendix
% \section{Additional experimental results}

% Here, we present additional experimental results, which were not supplied in the main text.

%%       \cite{<label>} ==> [1]

%%
%Example citation, See \cite{lamport94}.

  % \bibliographystyle{elsarticle-num} 
  % \bibliography{bib}



\end{document}